\documentclass{article}
\usepackage[utf8]{inputenc}

\usepackage{amsmath}
\usepackage{amssymb}
\usepackage{amsthm}
\usepackage{bm}

\usepackage[round]{natbib}
\usepackage{url}

\usepackage{graphicx}
\usepackage{caption}
\usepackage{subcaption}

\usepackage{floatrow}
\newfloatcommand{capbtabbox}{table}[][\FBwidth]

\usepackage{multirow}
\usepackage{adjustbox}

\usepackage[table,xcdraw]{xcolor}
\usepackage{hyperref}
\usepackage{makecell}

\usepackage{enumitem}
\usepackage{dsfont}
\usepackage{adjustbox}

\newtheorem{theorem}{Theorem}[section]
\newtheorem{lemma}[theorem]{Lemma}
\newtheorem{proposition}[theorem]{Proposition}
\newtheorem{remark}{Remark}
\newtheorem{definition}{Definition}[section]

\newcommand{\ddt}{\frac{d}{dt}}

\newcommand{\bX}{\mathbf{X}}
\newcommand{\Xt}{\mathbf{X}_t}
\newcommand{\Xj}{\mathbf{X}_j}
\newcommand{\vj}{\mathbf{v}_j}
\newcommand{\R}{\mathbb{R}}
\newcommand{\Rp}{\mathbb{R}^p}

\newcommand{\Ct}{\mathcal{C}_{t}}

\newcommand{\indep}{\perp \!\!\! \perp}

\DeclareMathOperator*{\argmin}{arg\,min}

\usepackage{color}

\title{Caliper Synthetic Matching: Generalized Radius Matching with Local Synthetic Controls}
\author{Jonathan Che \thanks{Department of Statistics, Harvard University} \and Xiang Meng \footnotemark[1] \and Luke Miratrix \thanks{Harvard Graduate School of Education}}

\begin{document}

\maketitle

\begin{abstract}
Matching promises transparent causal inferences for observational data, making it an intuitive approach for many applications.
In practice, however, standard matching methods often perform poorly compared to modern approaches such as response-surface modeling and optimizing balancing weights.
We propose Caliper Synthetic Matching (CSM) to address these challenges while preserving simple and transparent matches and match diagnostics.
CSM extends Coarsened Exact Matching \citep[CEM; ][]{iacus2012causal} by incorporating general distance metrics, adaptive calipers, and locally constructed synthetic controls.
We show that CSM can be viewed as a monotonic imbalance bounding \citep[MIB; ][]{iacus2011multivariate} matching method, so that it inherits the usual bounds on imbalance and bias enjoyed by MIB methods.
We further provide a bound on a measure of joint covariate imbalance.
Using a simulation study, we illustrate how CSM can even outperform modern matching methods in certain settings, and finally illustrate its use in an empirical example.
Overall, we find CSM allows for many of the benefits of matching while avoiding some of the costs.
\end{abstract}

\section{Introduction}

The ``Jovem de Futuro'' (Young of the Future) program of Brazil aimed to improve education quality by providing strategies and, for those who made certain targets, monetary support to schools in Rio de Janeiro and Sao Paulo, Brazil.
Education reform efforts like this are common, and when they occur various stakeholders are typically anxious to understand if such a large commitment of resources actually improved circumstances on the ground.

In the case of Jovem de Futuro, researchers ensured the ability to answer this question by taking a step beyond what is typically seen in practice: they recruited schools and randomized them into treatment and control, thus enabling direct and effective estimation of impacts by comparing two groups that were otherwise similar, save for treatment \citep{barros2012impacto}.
But getting agencies to implement a randomized controlled trial is often a hard sell: RCTs require supports beyond those of installing a program in the first place.
When, as is common in practice, there is no randomization to ensure one has a collection of comparable schools, researchers will instead turn to observational study methods to, in effect, construct a comparison group that can serve as a counterfactual for those schools who took treatment.

Matching provides a simple approach for such an endeavor.
In a basic setting, matching methods pair each treated unit with a similar control unit, producing a matched control sample that mirrors the treated sample in terms of observable covariates.
Under standard assumptions, the samples may then be analyzed as if treatment were randomly assigned.
The simplicity of this approach has led matching to be a popular method for observational causal inference \citep{imbens2004nonparametric, imbens2015causal}.

For conceptual clarity, the gold standard of matching methods is exact matching.
For each treated unit, exact matching finds a control unit with the same observed covariates.
We then end up with a control group that looks, in terms of the covariates, exactly like the treatment group, except for receipt of treatment.
In other words, exact matching perfectly balances the joint covariate distribution between the treated and matched control units, eliminating any potential bias due to associations between observable covariates and potential outcomes \citep{imai2008misunderstandings, rosenbaum1985bias}.
Exact matching also produces transparent matched datasets,
since the difference in outcomes between a treated unit and its exactly matched control is an unbiased (albeit noisy) estimate of the treatment effect for that unit.
This leads to the familiar statistical idea of averaging noisy observations to estimate a target estimand, e.g., the average treatment effect among all (or some subset) of the treated units.

In practice, however, exact matching is usually impossible.
Researchers have therefore developed a variety of methods for conducting principled causal inference without exact joint covariate balance.
For example, many matching approaches aim to balance the marginal distributions.
Researchers construct a matched dataset, check the marginal means of the matched sets, and then iteratively tweak and refine their model if the means are too different.
Direct balancing approaches \citep[e.g., ][]{hainmueller2012entropy, zubizarreta2015stable, ben2021balancing} improve on this ad-hoc procedure by directly targeting approximate balance on specified features of the marginal, or sometimes even joint, distribution.
Alternatively, one can use dimension reduction \citep[e.g., ][]{aikens2020pilot}, to make it easier to match ``locally.''
All of these forms of matching can be viewed as a preprocessing step \citep{ho2007matching} done before applying an outcome model, e.g., weighted linear regression, to the matched dataset; the final model can further adjust for remaining covariate imbalances.
Ideally the matching reduces the impact of model selection on the final estimates.

Matching and outcome modeling can also be done concurrently.
Semiparametric modeling approaches, such as doubly robust \citep{robins1994estimation, rotnitzky1998semiparametric}, double machine learning \citep{chernozhukov2018dml}, or targeted maximum likelihood estimation \citep{van2006targeted}, use model-assisted averages to target the estimand of interest, leading to provably efficient and unbiased estimates.
One can also simply directly model the outcome \citep[e.g., ][]{hill2011bayesian}, to similar effect.

The methods discussed above generally yield effective causal estimates by aiming for overarching objectives across all units, such as achieving marginal covariate balance or ensuring good model fit.
Even if they attain marginal covariate balance, however, they can fail to achieve joint balance, and this can lead to a poor result, as highlighted in \cite{iacus2011multivariate}.
Furthermore, such methods may not be particularly transparent, or can heavily depend on model specification, as noted in \cite{iacus2012causal}.
Exact matching, by contrast, targets joint balance, reduces model dependency and, ideally, increases transparency.
In this paper, we build on a body of literature that focuses on this original spirit of exact matching.
We provide four opinionated maxims for methods in this literature:
\begin{enumerate}
    \item Distances should be intuitive
    \item Matches should be local
    \item Match each unit as best you can in a way that you can monitor
    \item Estimates should be transparent
\end{enumerate}
Given our maxims, we then directly construct Caliper Synthetic Matching (CSM) to effectively and clearly address each of these ideas. 

The CSM method is a member of the Monotonic Imbalance Bounding \citep[MIB; ][]{iacus2011multivariate} class. MIB is a category of methods ensuring joint balance. When introducing Monotonic Imbalance Bounding (MIB), the authors indicated that while distance matching with a single scalar caliper does not constitute MIB, employing a distinct caliper for each covariate does. This paper explores this idea further by formalizing the scaled distance approach and examining its properties, specifically through a scaled distance metric. The scaled distance allows for specific control over each covariate, mitigating issues such as sensitivity to outliers inherent in similar methods such as Mahalanobis distance.

We also connect CSM to Coarsened Exact Matching \citep[CEM; ][]{iacus2012causal}.
CEM, a method within the MIB class, coarsens continuous covariates \footnote{E.g., dividing an age covariate into ranges such as 0-25, 25-50, 50-75, 75+} before exactly matching on these coarsened covariates.
We show how CSM keeps the MIB properties of CEM, but gives stricter bounds on worst-case bias.
Within this framework we also provide a novel theoretical result on the control of joint balance, along with an alternate derivation of bias control than in the original literature.

Another innovation of this paper is we, post matching, apply a variant of the synthetic control method \citep[SCM; ][]{abadie2010synthetic} to each treated unit in turn to refine the weights given to the control units matched to each treated unit.
The SCM method's goal is to create a single synthetic control unit for each treated unit that mirrors the treated unit’s characteristics and outcomes.
The synthetic control step further minimizes covariate imbalance over the matching step with respect to a scaled distance metric.
We show, via a Taylor expansion in a space defined by the used distance metric, how this leads to removing a linear bias term in the estimated impact in a manner akin to linear interpolation.
Our use of SCM diverges from standard SCM by allowing for a scaled distance metric, not necessarily using historical outcomes, and employing a predetermined $V$ matrix based on covariate-wise calipers.
We argue synthetic controls maintain the transparency of matching by enabling direct evaluation of counterfactual plausibility.

Our matching approach allows for a simple form of inference that accounts for the reuse of controls for different treated units.
Uncertainty estimation in matching methods is difficult, as underscored by Abadie and Imbens \citep{abadie2006large, abadie2008failure, abadie2011bias}, who illustrated the limitations of the bootstrap method and developed an asymptotically valid approach for the $M$-nearest neighbor estimator.
We follow \citet{ben2022synthetic} (also see \cite{keele2023hospital} and \cite{lu2023you}), using the resultant weights of the final matched control units along with a plug-in variance formula.
In particular, we use the residuals within matched clusters to estimate residual uncertainty, which then provides a variance estimator that takes into account the overall effective sample size imposed by the final control weights. 

The rest of the paper proceeds as follows.
In Section \ref{sec:background}, we provide background and motivate CSM with a toy example.
In Section \ref{sec:CSM}, we construct CSM in a stepwise manner to address the four maxims proposed above.
We then discuss the properties of CSM in Section \ref{sec:properties}, and provide our theoretical results on bias control.
In Section \ref{sec:simulation}, we conduct a variety of simulation studies to illustrate how CSM performs compared to other observational causal inference methods.
Finally, in Section \ref{sec:ferman-data-analysis} we work through the applied example of the Jovem de Futuro program, analyzing it as a within study comparison to demonstrate how CSM may be used in practice, and to assess its performance on real data.

\section{Background}
\label{sec:background}

\subsection{Setup}

Suppose we have $n$ independent and identically distributed observations, with $n_T$ treated units and $n_C$ control units.
For each unit $i$, let $Z_i \in \{0,1\}$ denote its binary treatment status, $Y_i \in \R$ denote its observed real-valued outcome, and $\mathbf{X}_i \equiv \{X_{1i}, \dots, X_{pi} \}^T \in \Rp$ denote its $p$-dimensional real-valued covariate vector.
We use the potential outcomes framework and denote the observed outcome for unit $i$ as $Y_i \equiv (1-Z_i) Y_i(0) + Z_i Y_i(1)$, for potential outcomes $Y_i(1)$ and $Y_i(0)$ under the stable unit treatment value assumption \citep{imbens2015causal}.
We make the standard conditional ignorability assumption:
\begin{align*}
    (Y(1), Y(0)) \indep Z \mid \bX,
\end{align*}
so that conditioning on the observed covariates is sufficient to identify the causal effect of $Z$.
Under a population sampling framework, we write $\epsilon_i \equiv Y_i(Z_i) - f_{Z_i}(\bX)$, where $f_0(\bX) \equiv E[Y(0) | \bX]$ and $f_1(\bX) \equiv E[Y(1) | \bX]$ are the true conditional expectation functions of the potential outcomes under control and treatment, respectively.
We write the set of all treated units' indices as $\mathcal{T} = \{i: Z_i=1\}$, the set of all control units' indices as $\mathcal{C} = \{i: Z_i=0\}$, and $t \in \mathcal{T}$, $j \in  \mathcal{C}$ as individual treated and control units respectively.  We denote the set of the indices of the control units matched to a treated unit $t \in \mathcal{T}$ as $\Ct = \{j \in \mathcal{C}: \text{ unit } j \text{ is matched to unit } t\}$.
Finally, we denote the size of a set $\mathcal{S}$ as $|\mathcal{S}|$.

In this paper, we will focus on estimating the sample average treatment effect on the treated (SATT):
\begin{align*}
    \tau = \frac{1}{n_T} \sum_{t \in \mathcal{T}} \left( Y_t(1) - Y_t(0) \right) ,
\end{align*}
where $n_T = |\mathcal{T}|$ denotes the number of treated units.

Under a population sampling framework, the SATT approaches the overall population average treatment effect on the treated (PATT) as the number of treated units increases.
While matching methods can easily be extended to estimate sample and population average treatment effects (SATEs and PATEs), we focus on the SATT to clarify key ideas and simplify exposition.

\subsection{Motivation: the spirit of exact matching}
\label{sec:toy}

To motivate the importance of locality and joint balance, we provide a pair of toy examples.
Figure \ref{fig:toy} plots the covariates $X_1$ and $X_2$ of control units $c_i$ and treated units $t_j$ for two scenarios.
The colors and contours visualize $f_0(\cdot)$, which takes on greater values within the innermost contour.
To conduct causal inference, we use the observed outcomes of the control units $c_i$ to impute the unobserved counterfactual outcomes of the treated units $t_j$.

\begin{figure}[t]
    \centering
    \begin{subfigure}[ht]{0.4\textwidth}
         \centering
         \includegraphics[width=\textwidth]{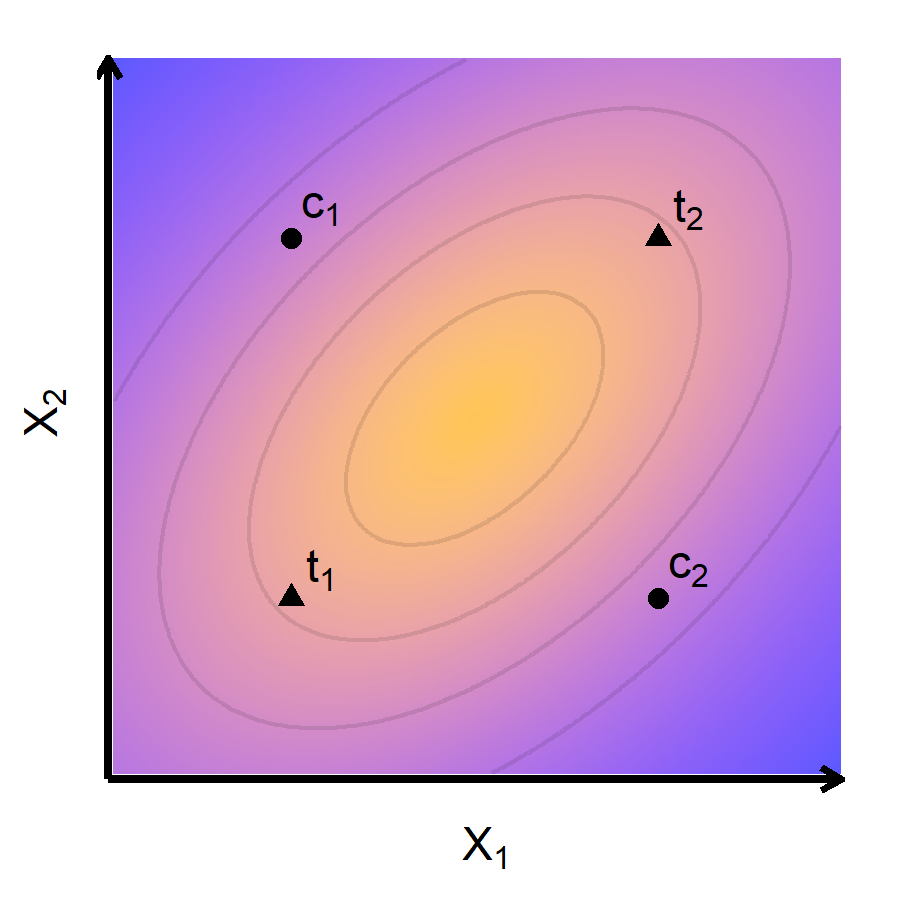}
         \caption{Toy Example 1}
         \label{fig:toy1}
     \end{subfigure}
     \hspace{5mm}
     \begin{subfigure}[ht]{0.4\textwidth}
         \centering
         \includegraphics[width=\textwidth]{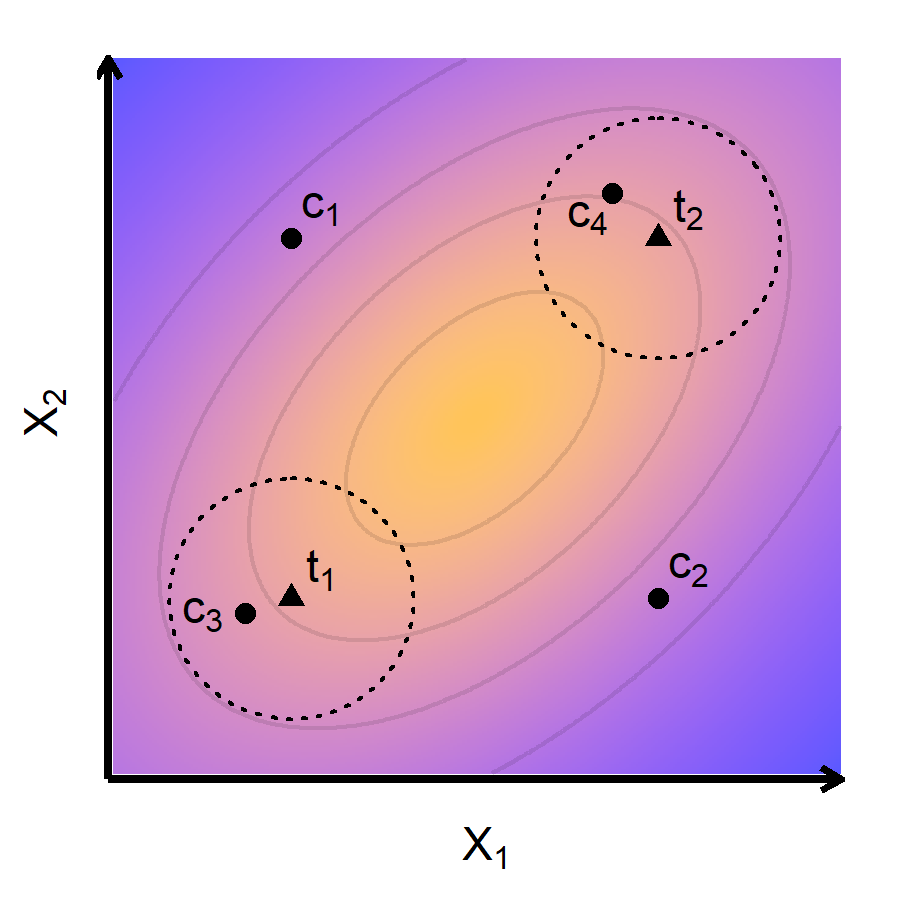}
         \caption{Toy Example 2}
         \label{fig:toy2}
     \end{subfigure}
    \caption{Toy examples demonstrating utility of locality and joint balance.}
    \label{fig:toy}
\end{figure}

For Example 1 (Figure \ref{fig:toy1}), accurate causal inference is not possible due to a lack of overlap: we cannot accurately impute outcomes for $t_1$ and $t_2$ because we do not observe any nearby evaluations of $f_0(\cdot)$.
Many causal inference methods, however, could fail to detect this problem.
For example, a matching or balancing approach that works with marginal distributions could exactly balance both $X_1$ and $X_2$ by assigning weights of 1 to all four units.
Similarly, an outcome model may fit a flat response surface to the observed outcomes of $c_1$ and $c_2$.
Because the usual marginal balance checks and outcome models would appear good, the analyst would be left unaware that these analyses significantly underestimate counterfactual outcomes, i.e., overestimate the SATT.
Approximate exact matching approaches would also fail here, as there are no control units close to the treated units.
The failure to find local matches, however, would alert the analyst to the presence of a problem.
Therefore, it is important to have a metric on how good controls each treatment can obtain.
In our method, we use nearest neighbor distance as this quality metric. 

Example 2 (Figure \ref{fig:toy2}) highlights the role of transparent local matches.
Because $f_0(\cdot)$ is smooth, using only the outcomes of local control units (i.e., controls $c_3$ and $c_4$ within the dotted circles) leads to accurate counterfactual estimates.
Local matches also produce transparent estimates;
unlike using a black-box outcome model, it is immediately clear how each control unit contributes to each treated unit's counterfactual estimate.
Finally, local matches encourage joint balance.
In Example 2, if we assigned weights of 1 to $c_1$ and $c_2$, and 0 to $c_3$ and $c_4$, we would have perfect marginal balance, but poor local match quality and a poor counterfactual.
It is worth sacrificing some marginal balance to upweight the closer units $c_3$ and $c_4$, since doing so greatly improves joint balance and the resulting treatment-effect estimates.

Having illustrated the importance of using distance matching to achieve local matches, we now turn to the selection of the distance metric.
The Mahalanobis distance has been commonly used in the literature, but it has drawbacks, including sensitivity to outliers and the curse of dimensionality. Outliers can distort the covariance matrix upon which Mahalanobis distance relies, leading to poor match quality.
Moreover, with an increase in the number of covariates, the effectiveness of Mahalanobis matching may diminish due to the curse of dimensionality, where units in high-dimensional spaces appear far apart, complicating the search for suitable matches.

To partially address these issues, we propose a more flexible metric, the scaled distance, which permits analysts to specify the desired tightness of matching for each covariate in turn.
Regions of low overlap can then be handled via an adaptive caliper method as detailed in Section \ref{sec:avoid}.
The challenge of dimensionality is also reduced, given prior knowledge on the relative importance of the covariates, as covariates of interest can be scaled to achieve higher level of control; see Section \ref{sec:dists}.

Lastly, while these toy examples demonstrate the benefits of adhering to the "spirit of exact matching" by finding local matches to enhance causal estimates and identify potential biases, it is important to acknowledge that in some, possibly even many, situations, focusing on joint balance in this manner may be unnecessary or even detrimental.
For example, if the conditional expectation function is additive in its covariates, i.e., $f_0(\bX) = \sum_{j=1}^p g_j(X_j)$ for functions $g_j(\cdot):\R \to \R$, marginal balance suffices \citep{zubizarreta2015stable}.
In such settings, attempting to balance the full joint covariate distribution is excessively conservative, as it protects estimates from bias that does not exist, namely, bias due to imbalances in  covariate interactions.

\subsection{Related work}
\label{sec:related}

Matching methods have a long history in observational causal inference \citep{stuart2010matching}.
Rather than attempt an exhaustive review, we briefly trace how these methods have operationalized the spirit of exact matching over time and provide further details in Section \ref{sec:CSM} as we develop the method introduced in this paper.

Early work in matching incorporated locality via nearest-neighbor \citep{rubin1973matching}, caliper \citep{cochran1973controlling, rosenbaum1985constructing}, and radius matching \citep{dehejia2002propensity} approaches.
These methods were typically combined with dimension reduction, e.g., via propensity scores \citep{rosenbaum1983central}, to circumvent the challenge of near-exact matching with multiple covariates, though some approaches, e.g., Mahalanobis distance matching \citep{rubin2000combining} and genetic matching \citep{diamond2013genetic}, directly operated on a scaled version of the original covariate space.
Also see \cite{aikens2020pilot}, who use set-aside data to fit a prognostic score, and then match on that score.
Regardless, to evaluate and improve the quality of a matched dataset, researchers would typically conduct iterative balance checks, revising their matching scheme if it led to poor marginal mean balance.

To circumvent the need for these iterative balance checks, \citet{iacus2012causal} introduced Coarsened Exact Matching (CEM), which we discuss in further detail in Section \ref{sec:close}.
CEM fixes a user-specified level of joint balance, operationalized by a given degree of coarsening, and returns the resulting sample.
By making local matches a primary rather than a secondary criterion, CEM enjoys the desirable transparency and joint balance properties of exact matching.

More recently, researchers have developed matching methods that flexibly learn notions of locality and similarity rather than using user-specified defaults.
Genetic matching \citep{diamond2013genetic} learns to more precisely match covariates that appear to be important for overall covariate balance, while ``almost-exact'' matching methods for discrete \citep{dieng2019interpretable, wang2021flame} and continuous \citep{morucci2020adaptive, parikh2022malts} data aim to more precisely (or exactly) match covariates that are more predictive of the potential outcomes.
Matching After Learning to Stretch \citep{parikh2022malts} minimizes predictive error over a hold-out training set to identify a similarity metric, and synthetic controls \citep{abadie2010synthetic} minimize the resulting predictive error with respect to a set of held-out pre-treatment outcome variables to obtain optimal variable importance weights in a time-series setting.

Work outside of matching has also noted the importance of locality in observational causal inference.
For example, \citet{abadie2021penalized} augments the popular synthetic controls methodology with a penalty for using control units far from the treated unit.
In a similar setting, \citet{ben2021augmented} tunes the extent to which synthetic controls may extrapolate away from the control units.
More generally, \citet{kellogg2021combining} explicitly trades off the bias from extrapolating beyond local matches with the bias from linearly interpolating between distant units.
While these approaches do not attempt to directly emulate exact matching, they highlight the use of local data for principled causal inference.

In another direction, one can directly optimize weights with respect to overall balance criterion, rather than engage in the iterative cycle of matching with a subsequent balance check on the result.
For three examples of many, see stable balance weights \citep[SBW; ][]{zubizarreta2015stable}, entropy balance weights \citep{hainmueller2012entropy},
or covariate balancing propentiy scores \citep{imai2014covariate}.

\section{Caliper Synthetic Matching}
\label{sec:CSM}

\citet{stuart2010matching} decomposes matching analyses into two phases.
In the design phase, researchers select a distance measure, use it to run a matching method, and diagnose the quality of the resulting matches.
In the subsequent analysis phase, researchers use the matched units to estimate treatment effects.

In this section, we propose a matching method that satisfies our set of aesthetic maxims delineated above.
We construct our method in a modular fashion;
at each stage, we increase the complexity of the method to improve its use of the data.
While we believe that the final proposed method simultaneously maximizes transparency and performance, in practice researchers may make different choices at each stage to limit complexity as necessary.

\subsection{Principle 1: Distances should be intuitive}
\label{sec:dists}

In the absence of exact matches, matching algorithms find control units as ``close'' as possible to their treated counterparts to improve joint covariate balance and reduce potential bias \citep{rosenbaum1985bias}.
One popular distance measure for ``closeness'' is propensity score distance:
\begin{align}
\label{eq:ps}
    d^{(e)}(\bX_i, \bX_j) = |e(\bX_i) - e(\bX_j)|,
\end{align}
where the propensity score $e(\cdot) = P(Z_i = 1 \mid \bX_i)$ represents the probability that a unit is treated, given its covariates \citep{rosenbaum1983central}.

While matching units with similar propensity scores leads to principled causal estimates, it does not construct intuitive matched sets  \citep{king2019propensity}.
Propensity score distance is not formally a distance metric on $\Rp$,\footnote{For example, $d^{(e)}(\bX_i, \bX_j) = 0$ does not imply that $\bX_i = \bX_j$.
For a simple introduction to distance metrics on $\Rp$, see Appendix \ref{app:distmetrics}.}
so it can violate our natural understanding of ``closeness.''
For example, two units that are ``close'' in terms of propensity score distance, which simply means they have similar chances of getting treatment, may have very different covariate values.
If a researcher matches two such units, it can be unclear whether they should trust this match, fit a better propensity-score model, or find a closer propensity-score match.

Formal distance metrics provide a more natural approach for assessing similarity between units.
Researchers are generally familiar with the covariates in their data, 
so directly attaching a distance metric to the space of covariates builds upon existing intuitions.
Multidimensional distance metrics typically take the form of a scaled Euclidean (i.e., $L_2$) distance metric:
\begin{align}
\label{eq:l2dist}
    d^{(2)}_V(\bX_i, \bX_j) = \sqrt{(\bX_i - \bX_j)^T (V^T V) (\bX_i - \bX_j)}
\end{align}
or scaled $L_\infty$ distance metric:
\begin{align}
\label{eq:linfdist}
    d^{(\infty)}_V(\bX_i, \bX_j) 
    = \sup_{k = 1, \dots, p} |V (\bX_i - \bX_j)|_k,
\end{align}
for a given $p \times p$ symmetric positive definite matrix $V$.
The matrix $V$ scales the raw differences in the covariates and their two-way interactions.
For example, if $V_{11}$ were large, then the resulting distance metric would magnify, i.e., upweight, differences in the first covariate. We consider both scaled $L_2$ and scaled $L_\infty$ distance metrics in this paper.

A variety of scaled distance metrics have been proposed in the literature.
One popular scaled $L_2$ distance metric is Mahalanobis distance, which uses $V^T V = \Sigma^{-1}$, the inverse covariance matrix estimated from the control group \citep{rubin1980bias}.
Other $L_2$ approaches, such as genetic matching \citep{diamond2013genetic} restrict the scale matrix $V$ to be diagonal and directly optimize it.

In this paper, we set $V$ to a diagonal matrix determined by a covariate-wise caliper specified by a researcher, as formalized in Proposition \ref{prop:distmetriccal}.
A covariate-wise caliper, $\pi_k$, represents the maximum allowable difference in the $k$-th covariate for matching purposes.
When the $L_\infty$ metric $d_V^{\infty}$ is set to $c$, the $k$-th covariate difference between the two units is at most $c \pi_k$.
The is equivalent to scaling each covariate $k$ by $V_{kk} = 1/\pi_k$ and then  matching using the $L_\infty$ distance with threshold $c$.
Generally, we advocate setting $\pi$ so that any pair of units with distance bounded by $c=1$ is considered to be a ``reasonably similar'' comparison.

Using $d_V^\infty$ with a covariate-wise caliper is strongly connected to Coarsened Exact Matching (CEM) \citep{iacus2012causal}, as the distance metric provides researchers with the desired degree of covariate-wise control.
In CEM, continuous variables are initially coarsened into discrete bins; for example, an age variable may be segmented into bins of 0-25, 25-50, 50-75, and 75-100 years.
Observations are then precisely matched based on these coarsened covariates.
By setting a covariate-wise caliper to 12.5, we allow a treated unit to be matched with a control unit if their age difference is no more than 12.5 years.
This approach generalizes the concept of coarsening, as coarsening, given the above bins, would not match a 26-year-old with a 24-year-old, but our method could. Further details are discussed in Section \ref{sec:compCEM}.

\subsection{Principle 2: Matches should be local}
\label{sec:close}

Section \ref{sec:toy} argued the importance of local matching. Here we discuss ways of achieving local matching by mainly revisiting the literature.

Given a chosen measure of ``closeness,'' there are many ways to select the closest matched units.
One popular approach is nearest-neighbors matching, where each treated unit is matched with the control unit(s) closest to it.
This can be done either greedily for each treated unit \citep{rubin1973matching} or optimally over all treated units \citep{rosenbaum1989optimal}.
Nearest-neighbors matching can be conducted after some preprocessing steps, such as learning an optimal distance metric \citep{diamond2013genetic, parikh2022malts}.

While nearest-neighbors approaches are intuitive, they can quietly fail to achieve covariate balance.
While nearest-neighbors approaches guarantee that each treated unit is matched with its closest control, the closest control may still be quite far off.
Large distances between treated units and their matched controls, i.e., low match quality, can lead to poor joint covariate balance.

To combat this problem, many methods apply calipers coupled with nearest-neighbor matching.
Calipers are a distance $c$ beyond which matches are forbidden.
Calipers modify the distance metric as:
\begin{align*}
    d(\bX_i, \bX_j) = 
    \begin{cases}
        d(\bX_i, \bX_j) &\text{if } d(\bX_i, \bX_j) \leq c \\
        \infty &\text{if } d(\bX_i, \bX_j) > c
    \end{cases}
\end{align*}
Using calipers, nearest-neighbor approaches can avoid problems associated with poor match quality, using the closest matches only if they are ``close enough.'' The simplest form of nearest neighbor matching does not, however, take full advantage of data rich areas of the sample: if there are many controls, using all of them rather than the strictly closest could improve precision.

Calipers are frequently applied to the propensity score, rather than to covariate distance.
Unfortunately, units that are local with respect to assignment probability may not be actually local with respect to their covariates; see \citet{king2019propensity}.

The distance metric and the caliper are quite connected; in particular, if we scale our distance metric, we can simply scale the caliper by the same amount.
This means $c=1$ indicates a caliper of one unit on the distance metric, which under the infinity-norm means no covariate differs by more than one ``width'' as defined by $\pi$.
Thus, without loss of generality, we generally assume an initial value of $c = 1$ hereafter, unless otherwise noted.

In this paper, we directly use all control units within a given caliper of each treated unit, which maximizes the number of local matches used to produce causal estimates.
We allow matching with replacement, meaning a control unit close to multiple treatment units will match to all of then.
This matching approach is known as radius matching \citep{dehejia2002propensity} with replacement, though previous proposals mostly use radius matching on propensity score distances.
We relegate discussion of caliper selection to Appendix \ref{app:caliperchoice}.

\subsection{Principle 3: Match each unit as best you can in a way that you can monitor}
\label{sec:avoid}

Matching with a fixed caliper could lead to a substantial fraction of the treated units being dropped, which can significantly change the target estimand. 
In particular, dropping difficult-to-match treated units, while potentially improving the quality of the resulting estimate, changes the estimand from the SATT to the feasible sample average treatment effect (FSATT):
$$\tau_\mathcal{F} = \frac{1}{|\mathcal{F}|} \sum_{t \in \mathcal{F}} Y_t(1) - Y_t(0),$$
where $\mathcal{F}$ denotes the set of indices of treated units with at least one control unit within $c$ $d_V^{(\infty)}$ units, i.e., $\mathcal{F} = \{t \in \mathcal{T}: \exists \ j \in \mathcal{C} \text{ with } d_V^{(\infty)}(\Xt, \Xj) \leq c\}$.
The FSATT can differ from the SATT if the excluded units have systematically different treatment impacts than the kept units.

Instead of shifting the estimand based on a selected caliper, we propose instead assigning each treated unit $t$ an adaptive caliper $c_t$, with
\begin{equation}
    c_t = \max \{c, \alpha d_t\} \mbox{ with } d_t \equiv \min_{j:Z_j=0} d(\bX_t, \bX_j),
\end{equation}
where $c$ is our global minimum caliper that we default to 1 given our distance metric, $d_t$ is the distance between unit $t$ and its nearest control-unit neighbor \citep{dehejia2002propensity}, and $\alpha \geq 1$ is an inflation factor that allows for catching all control units that are similarly close to the closest unit.

The adaptive caliper guarantees all treated units are matched to at least one control.
The floor $c$ value allows for taking advantage of data-rich environments; for treated units with many controls, we will not shrink to an overly small caliper but instead keep all those controls measured as ``reasonably similar'' as a comparison group.

In data-rich contexts, the adaptive caliper may also be selected so that the resulting matched sets work well with synthetic controls (introduced below), e.g., we can set $c_t$ to be the smallest caliper such that treated unit $t$ has $p+1$ within-caliper controls, where $p$ is the dimension of the covariate space. 

In principle, with adaptive calipers our matched dataset allows for direct estimate of the SATT as all treated units are preserved.
In practice, some treated units may have very poor matches, which would be indicated by large $c_t$ values.
Because we have a direct measure of match quality via $c_t$, we can monitor the impact of these poor matches on our overall estimand, and possibly deliberately trim based on our diagnostics.

An important step in any matching procedure is to assess the resulting matches.
Classic approaches would be to conduct balance checks by comparing marginal covariate distributions between the treated and control groups (usually people will look at standardized mean differences).
Such marginal balance checks may reveal significant departures from joint balance, but cannot confirm when joint balance is approximately achieved, as demonstrated in Toy Examples 1 and 2.
Checking low-dimensional summaries of joint balance can also fail to assess overlap or identify subsets of the treated units for which it may be easier or more difficult to estimate treatment effects.

Using a distance-metric caliper, on the other hand, directly ensures good covariate balance (e.g., see Proposition \ref{prop:wass}), if the caliper is fixed.
We can extend this idea to using the unit-specific caliper values $c_t$ across all treated units to assess the estimate-estimand tradeoff, creating balance-sample-size frontier plots \citep{king2017balance} to show how dropping poorly matched treated units (i.e., those with large $c_t$) affects both potential bias and the SATT estimate for the remaining sample.
Also see \citet{aikens2020pilot, aikens2023assignment} for other approaches to making diagnostic plots.
We can also summarize characteristics of units with different $c_t$ values to better understand regions with poorer overlap.

\subsection{Principle 4: Estimates should be transparent}

Given matched units from the design phase of a matching analysis, the final step is to produce an estimate.
With high-quality matches, the SATT may be, in principle, estimated with a simple average:
$$\hat{\tau}^{avg} = \frac{1}{n_T} \sum_{t \in \mathcal{T}} \left( Y_t - \frac{1}{|\Ct|} \sum_{j \in \Ct} Y_j \right).$$
This is quite clear, but when matching is imperfect, it can be much lower performing than other methods that do further adjustment.

In this paper, instead of using the average weight $\frac{1}{|\Ct|} $ for each matched control,  we calculate weights with the synthetic control method \citep[SCM; ][]{abadie2010synthetic} within the matched sets.
For each treated unit $t$, we find convex weights, i.e., weights $w_{jt} \geq 0$ that sum to 1, for its matched control units that minimize covariate imbalance as measured by a scaled distance metric $d_V(\cdot, \cdot) = d_V^{(2)}(\cdot, \cdot)$ or $d_V^{(\infty)}(\cdot, \cdot)$:
\begin{align*}
    \mathbf w_t := \argmin_{\{w_{jt} : j \in \Ct\}} 
        &\hspace{2mm} d_V(\Xt, \sum_{j \in \Ct}w_{jt} \Xj) \\
    \text{s.t. } 
        &\sum_{j \in \Ct} w_{jt} = 1 \\
        &0 \leq w_{jt} \leq 1 \text{ for } j \in \mathcal{C}_t,
\end{align*}
where we have written the weight for control unit $j$ associated with treated unit $t$ as $w_{jt}$.
The ``synthetic control'' unit for treated unit $t$ gets covariates $\sum_{j \in \Ct}w_{jt} \Xj$ and outcome $\sum_{j \in \Ct}w_{jt} Y_j$,
and the ATT estimate is taken as a simple difference-in-means between the outcomes of the treated units and their synthetic controls:
\begin{align*}
    \hat{\tau}_t^{SC} = \frac{1}{n_T} \sum_{t \in \mathcal{T}} \big(Y_t - \sum_{j \in \Ct} w_{jt} Y_j \big).
\end{align*}
See Appendix~\ref{app:scm} for further detail.

Our approach deviates from the standard SCM setup in a few ways.
First, while \citet{abadie2010synthetic} introduces SCM as a quadratic programming problem using scaled $L_2$ distance, we also allow SCM to be implemented as a linear programming problem using scaled $L_\infty$ distance.
Second, while synthetic controls are typically used in time-series settings with past outcomes as additional covariates, we directly apply them in our setting without focus on past outcomes.
Finally, the SCM introduced in \citet{abadie2010synthetic} includes an outer optimization to learn an ``optimal'' scaling matrix $V$, whereas we simply use the $V$ matrix implied by the given covariate-wise caliper.
This follows other approaches to the SCM. See, e.g., \cite{ben2021augmented}.

Using synthetic controls in the analysis phase provides two primary benefits.
First, synthetic controls are highly transparent. 
Because synthetic controls explicitly produce a counterfactual for each treated unit, researchers can directly check whether each counterfactual seems reasonable. 

Second, as we prove, synthetic controls naturally reduce bias within calipers as they are mathematically equivalent to linear interpolation.
While linear interpolation over long distances can lead to bias, interpolating over short distances, such as within a caliper, typically improves results due to the local linearity of smooth outcome functions.

\subsection{The full method}

To recap, we propose matching by first selecting an interpretable distance metric that can be used to assess similarity between units based on their covariate profiles.
This metric should be interpretable and justifiable by those using it.

Once a metric is in place, identify sets of local controls for each treated unit, possibly adaptively adjusting the caliper (radius) defining locality unit by unit.
Finally, build synthetic comparison units by using a variant of the synthetic control procedure for each unit to obtain further reductions in bias.
The final estimate is simply the average of the pairwise differences between treated and synthetic control.
Finally, run some diagnostics, examining measures such as how the estimated impact shifts with the trimming of hard-to-match units.

We denote the use of synthetic controls within adaptive $L_\infty$ calipers as Caliper Synthetic Matching (CSM).
We generally use a scaled $L_\infty$ distance for its simplicity, its connections to exact matching, and its connections to CEM.
Adapting the caliper enables clear diagnostic plots of the estimate-estimand tradeoff and the synthetic controls produce interpretable local bias corrections.

\section{Theoretical Properties}
\label{sec:properties}
\subsection{Monotonic Imbalance Bounding}
\label{sec:mib}

\citet{iacus2011multivariate} introduces the Monotonic Imbalance Bounding (MIB) class of matching methods.
MIB matching methods directly control covariate balance between the treated and matched control groups, independently for each covariate. 
As a result, MIB matching methods enjoy desirable properties such as bounded covariate imbalance and bounded estimation error, under reasonable assumptions.

Although matching using Mahalanobis distance alone is not MIB, matching using a distance-metric caliper is MIB as long as the caliper for each covariate may be tuned without affecting the caliper for any other covariate \citep{iacus2011multivariate}.
Any such covariate-wise caliper can be satisfied by a scalar caliper on an appropriately scaled distance metric.
For example, if $p=2$ and we want to ensure $|X_{t1} - X_{j1}| \leq 2$ and $|X_{t2} - X_{j2}| \leq 5$, we may define $V = \begin{bmatrix} \frac{1}{2} & 0 \\ 0 & \frac{1}{5} \end{bmatrix}$ so that restricting $d^{(\infty)}_V(\Xt, \Xj) \leq 1$ satisfies the desired caliper.
By the triangle inequality, the bound on the distance translates to bounds on the individual covariates.
This idea is formalized by Proposition \ref{prop:CSM-is-MIB}, which shows that matching using a distance-metric caliper is a member of the MIB class.

Denote the weighted mean for covariate $k$ for the treated and control units after matching as $\bar{X}^{\mathbf{w}}_{Tk}$ and $\bar{X}^{\mathbf{w}}_{Ck}$, respectively, where the weight is obtained from matching.
More detailed definitions can be found in Appendix \ref{app:mib}.
We then have:

\begin{proposition}
\label{prop:CSM-is-MIB}
    Given covariatewise caliper $\boldsymbol{\pi} \in \Rp_{+}$ and scaling matrix $V = diag\{\frac{1}{\boldsymbol{\pi}}\}$, we have 
    \begin{enumerate}[label=(\alph*)]
        \item $d^{(2)}_V(\Xt, \Xj) \leq 1$ for all  $t \in \mathcal{M}_T, j \in \Ct$  $\implies |\bar X^{\mathbf w}_{T k}- \bar X^{\mathbf w}_{C k}| \leq \pi_k \text{ for } k=1,\dots,p $
        \label{prop:dmca}
        \item $ d^{(\infty)}_V(\Xt, \Xj) \leq 1 $ for all  $t \in \mathcal{M}_T, j \in \Ct$   $\implies  |\bar X^{\mathbf w}_{T k}- \bar X^{\mathbf w}_{C k}| \leq \pi_k \text{ for } k=1,\dots,p $
        \label{prop:dmcb}
    \end{enumerate}
\end{proposition}

\begin{proof}
    See Appendix \ref{app:mib}.
\end{proof}

 Proposition \ref{prop:CSM-is-MIB} not only shows that changing $\pi_k$ for one variable does not affect the imbalance on the other variables, but also shows that covariate balance is controlled by $V$.
 Again, note that the choice to bound $d^{(2)}_V(\Xt, \Xj)$ and $d^{(\infty)}(\Xt, \Xj)$ by 1 is arbitrary, e.g., bounding $d^{(\infty)}_V(\Xt, \Xj) \leq 1$ is equivalent to bounding $d^{(\infty)}_V(\Xt, \Xj) \leq c$ for $c>0$ and $V = diag\{\frac{c}{\pi}\}$.
 
 With adaptive calipers, CSM is MIB for the feasible treated units within $\mathcal{F}$ but $d_V(\Xt, \Xj) \leq 1$ does not hold for other treated units, and thus the bound does not apply.
 That said, the adaptive calipers still provide transparency about the extent to which CSM leaves the MIB class in terms of number of units and the amount they deviate as recorded by their caliper sizes.
 In the subsequent subsections, we focus on cases when all treated units have controls within a set caliper and thus are feasible, i.e., $\mathcal{M}_T = \mathcal{F}$, to focus on several properties of the MIB matching methods.

\subsection{Bounded joint covariate imbalance}

Having shown the bounded marginal balance, we now show how distance-metric caliper matching methods control joint covariate imbalance. Write the empirical joint covariate distributions of the treated and control units as:
\begin{align*}
    f_T(\mathbf{x}) 
    &\equiv \frac{1}{n_T} \sum_{t \in \mathcal{T}} \delta_{\Xt}(\mathbf{x}) \\
    f_C(\mathbf{x}) 
    &\equiv \frac{1}{n_T} \sum_{t \in \mathcal{T}} \sum_{j \in \Ct} w_{jt} \delta_{\Xj} (\mathbf{x}),
\end{align*}
where $\delta_{\bX}(\cdot)$ represents a Dirac delta function at $\bX$, i.e., $\delta_{\bX}(\mathbf{y}) = 1$ if $\mathbf{y} = \mathbf{X}$ and $0$ otherwise.
These empirical distributions are simply weighted sums of point masses located at the units' values.

To demonstrate how distance-metric calipers control the difference between $f_T$ and $f_C$, we use Wasserstein distance.
Formally, the $q$-Wasserstein distance between probability distributions $P$ and $Q$ with distance metric $d(\cdot, \cdot)$ is:
\begin{align*}
    \mathcal{W}_q(P, Q) = \inf_{\pi \in \Pi(P,Q)}  E_\pi\big[ d(\bX, \mathbf{Y})^q \big]^{1/q},
\end{align*}
where the infimum is over all couplings (joint distributions) $\pi$ that have marginal distributions of $P$ and $Q$ ($Pi(P,Q)$).
More intuitively, Wasserstein distance is also known as ``earth-mover's distance'':
viewing $P$ and $Q$ as piles of soil each with total mass 1, $W_q(P, Q)$ measures the minimum ``cost'' to move soil to make $P$ distribute as $Q$ (or vice-versa), where the ``cost'' of moving $k$ units of soil from $\mathbf{x}$ to $\mathbf{y}$ is $k \cdot d(\mathbf{x}, \mathbf{y})$.
Wasserstein distance therefore uses a distance metric between points in $\Rp$ to measure distance between full probability distributions on $\Rp$.

With this notation in hand, Proposition \ref{prop:wass} shows that radius matching bounds the Wasserstein distance between $f_T$ and $f_C$.

\begin{proposition}
\label{prop:wass}
    For $c > 0$, and a matching method, we have, for both $\ell = 2$ and $\infty$: 
    \begin{align*}
        \text{For all } t \in \mathcal{T},
        j \in \Ct ,  d^{(l)}_V(\Xt, \Xj) \leq c
        \implies
        \mathcal{W}^{(l)}_q(f_T, f_C) \leq c
    \end{align*}
\end{proposition}

\begin{proof}
    See Appendix \ref{app:wass}.
\end{proof}

We make a few remarks about Proposition \ref{prop:wass}.
First, while the notation is technical, the intuition is straightforward.
Distance-metric calipers control how far each treated unit's covariates can be from its matched controls' covariates.
Since $f_T$ and $f_C$ are weighted sums of the point masses associated with these covariates, the calipers must also control the distance between $f_T$ and $f_C$.
For an exact caliper $c=0$, Proposition \ref{prop:wass} simply states that exact matching guarantees that $f_T$ coincides with $f_C$.
In practice, however, reducing the caliper size to $c=0$ drops all of the treated units, leaving the proposition to be vacuously true; $c$ must be chosen with the estimate-estimand tradeoff in mind.
To the best of our knowledge, this is the first theoretical result for a matching method's bound on joint balance.
\citet{iacus2011multivariate} did show that CEM improved joint balance, but only through simulation. 

The Wasserstein distance depends on the chosen distance metric.
In the proposition, $\mathcal{W}^{(l)}_q(f_T, f_C)$ is dependent on $d_V$.
Thus, if our $V$ changes, the Wasserstein distance also changes.
However, the core idea of the proof remains intact: the bounds on the covariates due to matching ensure a joint balance between the covariates of the matched controls and treatments, thereby providing local control and helping to avoid those adverse situations discussed in Section~\ref{sec:toy}.

Control over joint covariate imbalance naturally implies control over marginal imbalances.
For example, Proposition \ref{prop:meanbd} shows how distance-metric calipers also bound the distance between the ($p$-dimensional) empirical weighted covariate means of the treated and matched control units, respectively denoted $\bar{\bX}_T$ and $\Bar{\bX}_C$.

\begin{proposition}
\label{prop:meanbd}
    For $c > 0$ and a matching method, we have, for both $\ell = $ 2 and $\infty$: 
    \begin{align*}
         \text{For all } t \in \mathcal{M}_T, j \in \Ct ,  d^{(l)}_V(\Xt, \Xj) \leq c 
        \implies d_V^{\ell}(\bar{\bX}_T, \bar{\bX}_C) \leq c
    \end{align*}
\end{proposition}
\begin{proof}
    See Appendix \ref{app:meanbd}.
\end{proof}
The above proposition is in effect a joint version of Proposition~\ref{prop:dmcb}, which gives bounds on the covariates in turn.

Lastly, we note that other distance metrics between $P$ and $Q$ could be chosen. For instance, \cite{iacus2011multivariate} uses the $L_1$ distance, demonstrating that Coarsened Exact Matching reduces the joint imbalance between the control and treatment groups using an empirical example.
We use $L_1$ distance for its computational simplicity and use the Wasserstein distance because its properties facilitate easier proof construction, as it inherently utilizes a distance metric.

In summary, distance-metric calipers enable precise control of joint covariate imbalance.
While these bounds are not necessarily small in practice, they guarantee that observed imbalance cannot be too great, even in the worst case.
This leads to a variety of desirable properties which we illustrate in the following sections.

\subsection{Bounded bias}
\label{sec:biasbd}

Because MIB matching methods bound the distance between the covariates of matched units, they naturally bound the distance between smooth functions $f:\Rp \to \R$ of those covariates as well.
Write the control potential outcome for unit $i$ as $Y_i(0) = f_0(\bX_i) + \epsilon_i$.
Then assuming that $f_0(\cdot)$ is smooth (i.e., Lipschitz) immediately bounds the bias of any FSATT estimate produced by a method using a distance-metric calipers. (See Appendix \ref{app:lipschitz} for technical details about Lipschitz functions in $\Rp$.)

\begin{proposition}
\label{prop:biasbd_lip}
Suppose $f_0: \Rp \to \R$ is Lipschitz$(\lambda)$ with respect to $d_V(\cdot, \cdot)$ = $d^{(2)}_V(\cdot, \cdot)$ or $d^{(\infty)}_V(\cdot, \cdot)$.
Then for a matching procedure such that for all $t$, $d_V(\Xt, \Xj) \leq c$ for all $j \in \Ct$:
\begin{equation*}
    \big|E[\tau - \hat{\tau}] \big| \leq \lambda c.
\end{equation*}
\end{proposition}
\begin{proof}
    \begin{align*}
        \big| E &[\tau - \hat{\tau} ] \big| \\
        &= \bigg| \frac{1}{n_T} \sum_{t \in \mathcal{T}} 
            \Big( \sum_{j \in \Ct} w_{jt} \big(f_0(\Xt) - f_0(\Xj)\big) \Big) \bigg| \\
        &\leq \bigg| \frac{1}{n_T} \sum_{t \in \mathcal{T}} 
            \Big( \sum_{j \in \Ct} w_{jt} \lambda d(\Xj, \Xt) \Big) \bigg| \\
        &\leq \lambda c.
    \end{align*}
\end{proof}

Proposition \ref{prop:biasbd_lip} states that for distance-metric caliper matching methods, worst-case bias is proportional to the caliper size $c$.\footnote{Proposition \ref{prop:biasbd_lip} applies to a slightly different set of methods than Proposition 1 from \citet{iacus2011multivariate}, which proves a similar bias bound for MIB matching methods (see Appendix \ref{app:mib}).}
As in Proposition \ref{prop:wass}, setting $c=0$ shows that exact matching leads to unbiased estimates, but in practice we must use $c > 0$ to navigate the estimate-estimand tradeoff without dropping all of the treated units.

Of course, we rarely know the Lipschitz constant $\lambda$ in practice, so Proposition \ref{prop:biasbd_lip} does not generate empirical bias bounds.
Nonetheless, Proposition \ref{prop:biasbd_lip} illustrates how distance-metric calipers control bias.
If each treated unit is close to all of its matched controls in covariate space, its expected counterfactual outcome must be close to their expected outcomes, for reasonable (i.e., smooth) outcome functions.
As a result, any weighted average of the control units' outcomes cannot differ too much in expectation from the treated unit's true counterfactual outcome, regardless of the specific form of $f_0(\cdot)$. 

In practice, we can ideally get estimates that beat this bound by assuming local linearity and adjusting further for residual imbalance. We discuss this in the next section.

\subsection{Bias reduction from synthetic controls}
\label{sec:biasbdscm}

While using distance-metric calipers bounds bias, using synthetic controls within these calipers actively reduces bias.
Specifically, synthetic controls naturally remove linear bias by conducting local linear interpolation.
This means that if $f_0(\cdot)$ is linear, synthetic controls would completely eliminate bias if the weights achieved perfect balance on $X$.
For nonlinear $f_0(\cdot)$, bias will be reduced to higher-order nonlinear trends, which are, in general, less influential within tight calipers for smooth functions.
Proposition~\ref{prop:scbiasbd} more precisely shows how exact synthetic controls eliminate linear bias.

\begin{proposition}
\label{prop:scbiasbd}
Let $d_V(\cdot, \cdot)$ = $d^{(2)}_V(\cdot, \cdot)$ or $d^{(\infty)}_V(\cdot, \cdot)$.
Suppose $f_0: \Rp \to \R$ is differentiable and Lipschitz($\lambda$) with respect to $d_V$.
Then for a matching procedure such that for all $t$, $d_V(\Xt, \Xj) \leq c$ for all $j \in \Ct$,
if $\sum_{j \in \Ct} w_{jt} \Xj = \Xt$ for all $t \in \mathcal{T}$ (i.e., we have perfect balance locally for each treated unit):
\begin{equation*}
    \big|E[\tau - \hat{\tau}] \big| \leq o(c).
\end{equation*}
\end{proposition}
\begin{proof}
See Appendix~\ref{app:scbiasbd} for a proof based on a Taylor expansion with respect to the given distance metric.
The resulting expansion differs from the usual multivariate Taylor expansion, since it uses directional derivatives to better utilize the distance-metric calipers.
\end{proof}

In Proposition~\ref{prop:scbiasbd}, the $o(c)$ term represents higher-order terms which go to zero more quickly than does the caliper $c$ as $c$ shrinks toward zero.
Compare to the bias bound of $\lambda c$ \emph{without} the synthetic control step in Propostion~\ref{prop:biasbd_lip}: we have lost the $\lambda c$ term, leaving something smaller (asymptotically).
In practice, shrinking $c$ to zero drops all treated units, but the notation highlights how implementing synthetic controls within calipers takes advantage of local linearity.
In particular, while linear interpolation across large distances can lead to significant interpolation bias \citep{kellogg2021combining},
restricting the donor-pool units to be within a caliper distance from the treated unit reduces the impact of nonlinearities.

\subsection{Comparing CSM to CEM}
\label{sec:compCEM}

As discussed in Section \ref{sec:close}, radius matching methods have many similarities with coarsened exact matching (CEM) \citep{iacus2012causal}, which we celebrate by matching two of the three letters.
In particular, radius matching and CEM both possesses the imbalance-bounding and bias-bounding guarantees discussed in the previous sections.

To clarify the benefits of radius matching, we consider an example with two covariates $X_1$ and $X_2$, as in Figure \ref{fig:vs_cem}. The coarsening of each variable is made equal sized. We visualizes the CEM grid defined by the coarsening of the two covariates, 
along with an equivalently sized $L_\infty$ caliper around treated unit $t_1$.
\begin{figure}[t]
    \centering
    \includegraphics[width=\textwidth]{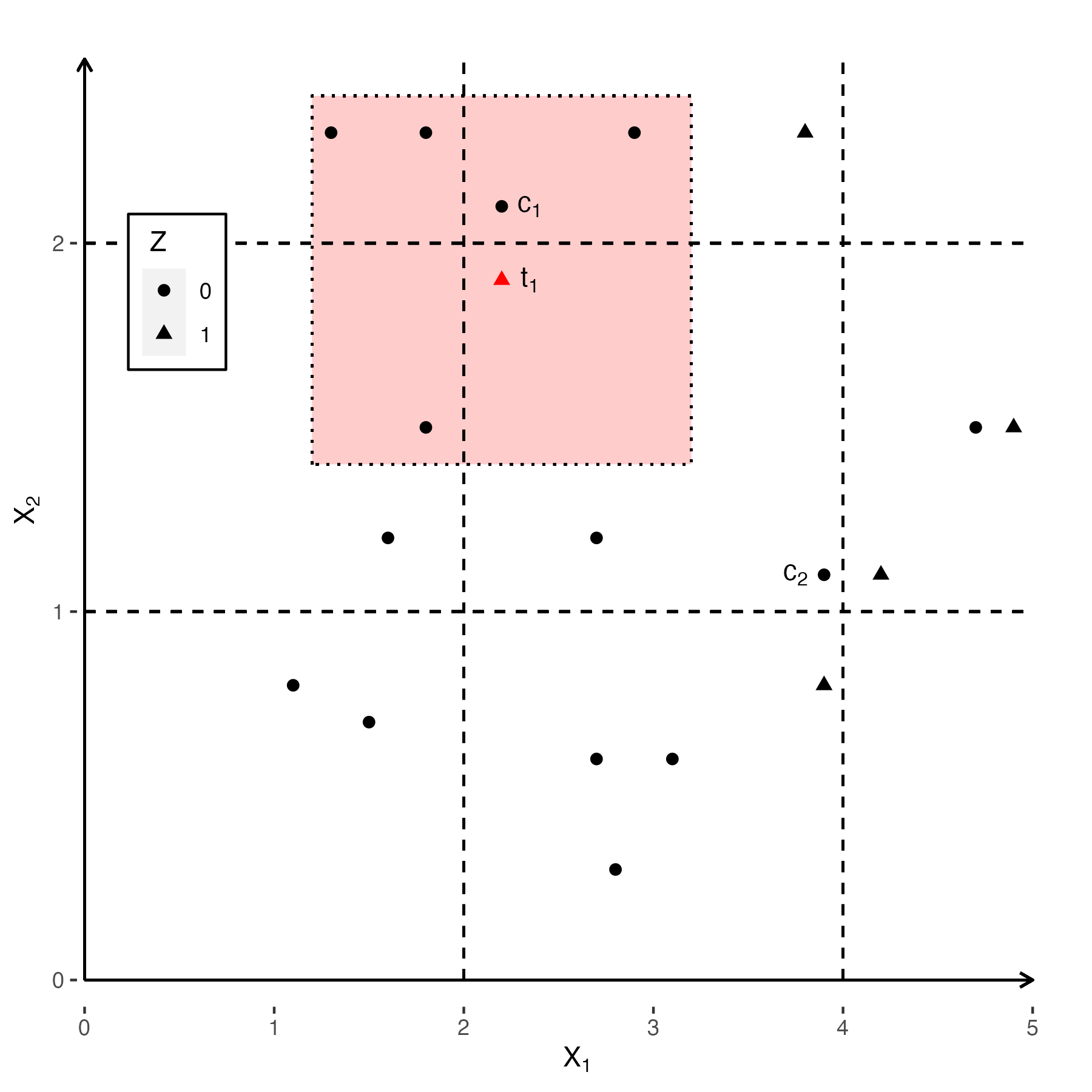}
    \caption{Comparison of coarsened exact matching with radius matching using $L_\infty$ calipers.
    Dashed gridlines represent CEM calipers.
    Shaded box represents scaled $L_\infty$ caliper around point $t_1$.}
    \label{fig:vs_cem}
\end{figure}
Because $t_1$ lies near a boundary defined by the covariate coarsening, CEM does not match $t_1$ to $c_1$, while the $L_\infty$ caliper does.
On the other hand, CEM matches $t_1$ to $c_2$ because they lie in the same caliper grid, even though the two units lie more than one $L_\infty$ unit apart from each other.
Figure \ref{fig:vs_cem} shows how, given a fixed caliper size, CEM guarantees that treated units lie within two $L_\infty$ units of their matched controls whereas the $L_\infty$ caliper guarantees the distance is no more than one unit.
By centering calipers on each treated unit, radius matching matches each treated unit to all nearby control units while guaranteeing imbalance and bias bounds that are twice as tight as those guaranteed by CEM. In other words, when each method is ``equally sized'' in terms of how large a treated unit's cachement area for possible control is, radius matching has tighter control of bias. 

As a result, we can present the following proposition, which highlights the improved guarantee of bias bounds. Specifically, it improves the upper bound of the bias for CSM compared to CEM.

\begin{proposition}
Suppose $f_0: \Rp \to \R$ is Lipschitz$(\lambda)$ with respect to $d_V(\cdot, \cdot)$ = $d^{(2)}_V(\cdot, \cdot)$ or $d^{(\infty)}_V(\cdot, \cdot)$, and $\boldsymbol{\pi} \in \Rp_{+}$. 
Then CSM with covariate-wise caliper $\pi / 2$ and CEM with covariate-wise coarsening $\pi$ have the same cachement size for each treated unit, but

\begin{equation*}
    \big|E[\tau - \hat{\tau}] \big| \leq \begin{cases}
             \lambda c  \text{ for CEM}  \\ 
             \lambda c/2  \text{ for CSM} 
        \end{cases} 
\end{equation*}
\end{proposition}

Naturally, there are tradeoffs for these improved bias bounds. Computationally, radius matching requires computing distances between each treated unit and the control units, an operation of order $n_t n_c$, unlike CEM which only requires a frequency tabulation of order $n_t+n_c$.
Non-uniform calipers are also slightly easier to implement via covariate coarsenings, though for many common covariates one can directly transform the covariate and use a uniform caliper to replicate a non-uniform caliper.\footnote{E.g., rather than non-uniformly coarsening income as \{ \$0-20k, \$20k-50k, \$50k-100k, \$100k+ \}, it may be more reasonable to log-transform the income covariate and use a uniform caliper to avoid, e.g., concluding that an individual earning \$20k is as far from an individual earning \$20.1k as they are from an individual earning \$50k.}
Overall, however, radius matching preserves the transparency and interpretability of CEM while significantly improving on its useful bias and imbalance bounding properties.

We also note that CSM offers the flexibility to select units that are challenging to match effectively, a feature that CEM lacks.
In CEM, after specifying the desired coarsening size, an analyst can only determine if a treated unit has found a match by checking if there is a control within the same stratum.
Conversely, with the adaptive caliper, we can evaluate the treated units based on the size of their adaptive calipers, which indicate the extent of their matching difficulty.
By contrast, CEM allows for a clean inferential strategy by viewing the bins as strata in a post-stratified experiment; with CSM inference is less clear cut, as we will discuss below.

\subsection{Bias-variance tradeoff}
\label{sec:bv}

Reducing caliper sizes reduces the potential bias of the resulting SATT estimate, as shown by Proposition \ref{prop:biasbd_lip}, but may increase the estimate's variance due to reducing the number of controls and, without adaptive calipers, the number of treated units.
To formalize this idea, we introduce the (conditional) mean-squared error for the SATT $\tau$:
\begin{align*}
    CMSE = E[(\hat{\tau} - \tau)^2 | \bX, \mathbf{Z}],
\end{align*}
where the expectation is conditioned on the observed covariates and treatment assignments.
Then, using a working model of the potential outcomes of $Y_i(z) = f( X_i ) + z \tau_i + e_i$, where $e_i$ is a 0-centered noise term, we can use standard algebraic manipulation to show that:
\begin{align}
    \hat{\tau} - \tau &= Bias + Error \label{eq:hat-tau-breakdown}
\end{align}
where 
\begin{equation}
\begin{split}
    Bias &= 
        \frac{1}{n_T} \sum_{t \in \mathcal{T}} \sum_{j \in \Ct} 
            w_{jt} \big( f_0(\Xt) - f_0(\Xj) \big)    \\
    Error &=
        \frac{1}{n_T} \sum_{t \in \mathcal{T}} (\epsilon_t - \sum_{j \in \mathcal{C}} w_{jt} \epsilon_j) . \label{eq:bias-variance}
\end{split}
\end{equation}

We then have
\begin{align*}
    CMSE &= Bias^2 + Var ,
\end{align*}
where, using our working model,
\begin{align}
\label{eq:true-variance}
    Var &=
        \frac{1}{n_T^2} \sum_{t \in \mathcal{T}} \sigma_t^2 +
        \frac{1}{n_T^2} \sum_{j \in \mathcal{C}} (\sum_{t \in \mathcal{T}} w_{jt})^2 \sigma_j^2 ,
\end{align}
with $\sigma_i$ representing the population sampling variance of unit $i$ conditional on its covariates $\bX_i$ \citep{kallus2020generalized}. 

If we assume homoskedasticity of the residuals within treatment arms (i.e., $\sigma_i = \sigma_C$ in the control group) and the conditions of Proposition \ref{prop:biasbd_lip} with caliper $c$, we can bound CMSE as:
\begin{align}
    CMSE &\leq 
        (\lambda c)^2 +
        \Big( \frac{\sigma^2_T}{n_T} + \frac{\sum_{j \in \mathcal{C}} (\sum_{t \in \mathcal{T}} w_{jt})^2\sigma^2_C}{n_T^2} \Big) \nonumber \\
        & = (\lambda c)^2 +   \frac{\sigma_T^2}{ESS(\mathcal{T})} +  \frac{\sigma^2_C}{ESS(\mathcal{C})} , \label{eq:cmsebd}
\end{align}
where the effective sample size (ESS) of a set $\mathcal{S}$ of units with weights $w_i$ is
\begin{align}
    ESS(\mathcal{S}) = \frac{(\sum_{i \in \mathcal{S}} w_i)^2}{\sum_{i \in \mathcal{S}} w_i^2}, \label{eq:ESS}
\end{align}
with $w_i = \sum_{t \in \mathcal{T}} w_{it}$ for $i \in \mathcal{C}$, and  $w_i = 1/n_T$ for $i \in \mathcal{T}$.

Equation \ref{eq:cmsebd} clarifies the relationship between caliper size and the bias-variance tradeoff.
For a fixed estimand, where we do not drop or add treated units as caliper size changes,increasing $c$ naturally exposes the resulting estimate to more bias.
Increasing $c$ also generally reduces variance by dispersing weight across more control units, increasing the effective sample size of the control group.

That said, in contexts where $|\mathcal{C}| \gg |\mathcal{T}|$, the overall variance of the average impact estimate could be dominated by the variance associated with the treated units (of order $1/ESS(\mathcal{T})$, which remains unchanged even if more control units are added (which shrinks $1/ESS(\mathcal{C}$.
To be more explicit, consider a single treated unit $t$ with matched controls $j \in \Ct$ all given uniform weights.
Here the variance associated with the treated unit is
$\sigma^2/1$ while the variance associated the controls is $\sigma^2/|\Ct|$ (assuming homoskedasticity).
Even if $\Ct$ were large, we still have the initial $\sigma^2$ term.
Once $\Ct$ reaches 4 or 5, increasing caliper size would have diminishing returns on variance reduction, suggesting that it may typically be better, in terms of CMSE, to use smaller calipers.

On the other hand, when matching with replacement, a single control $j$ may be assigned significant weight for multiple treated units $t$.
In these cases, $\sum_{t \in \mathcal{T}} w_{jt}$ could be greater than 1 for some control units, driving $ESS(\mathcal{C})$ down.
In other words, if the same controls get reused heavily enough, the variance associated with the control units may not be dominated by the variance associated with the treated units, even if the initial control pool is large.

\subsection{Variance Estimation}
\label{sec:var-est}
How to estimate standard errors for our matching method requires special attention.
Abadie and Imbens, through a series of papers (\citeyear{abadie2006large}, \citeyear{abadie2008failure}, \citeyear{abadie2011bias}), argued that the bootstrap method was inadequate and developed an asymptotically valid approach for the $M$-nearest neighbor estimator as the numbers of treatment and control units increase infinitely, with $M$ remaining fixed.
\cite{otsu2017bootstrap} introduced a weighted bootstrap methodology that relies on overlap of treated and control covariates.

We instead draw from the survey sampling literature \citep[e.g.,][]{potthoff1992equivalent} and plug in effective sample sizes and variance estimates into our Equation~\ref{eq:cmsebd}.
This has been successfully applied in various different types of weighting estimators in causal inference \citep{ben2022synthetic, lu2023you, keele2023hospital}.

Our approach differs from the Abadie and Imbens literature in several aspects: First, we utilize synthetic control weights instead of average $M$-nearest-neighbors weights.
Second, we condition on the treated units as defined by their covariates, assume a stochastic model for the units' outcomes, and make a working assumption of homoskedasticity.
Third, we operate in a scenario where the number of treatment units ($n_T$) is low, while the number of control units ($n_C$) is high.
As \cite{ferman2021matching} noted, the asymptotics in this setting differ from those Abadie and Imbens described, with $n_T$ increasing to infinity.

Our variance estimator is the following plug-in estimator:
\begin{align}
\label{eq:plug-in-variance-estimator}
    \widehat{Var}(\hat{\tau}) &= S^2 \left( \frac{1}{n_T} + \frac{1}{\text{ESS}(\mathcal{C})} \right)
\end{align}
where $S^2$ is an estimate of the homoskedastic residuals (assumed the same for treatment and control groups), and $\text{ESS}(\mathcal{C})$ follows \ref{eq:ESS}. 
In particular, $S^2$ is a pooled variance estimator, pooling across clusters: 
$$ S^2 = \frac{1}{N_C} \sum_t |\mathcal{C}_t| s^2_t \text{ with } N_C = \sum_t |\mathcal{C}_t|. $$
where each cluster's residual variance is
$$ s^2_t = \frac{1}{|\mathcal{C}_t| - 1} \sum_{j \in \mathcal{C}_t} e_{tj}^2 , $$
with
$$ e_{tj} = Y_{tj} - \bar{Y}_t, \text{ where } \bar{Y}_t = \frac{1}{|\mathcal{C}_t|} \sum_{j \in \mathcal{C}_t} Y_{tj} .$$
We drop all clusters $\mathcal{C}_t$ with $|\mathcal{C}_t| = 1$, including for the calculation of $N_C$.
Although we could use the weights from the synthetic control to prioritize ``good'' controls, we opt for uniform weights over all matched units as the key element is control unit similarity, which is not necessarily optimized by the synthetic step.

Our estimates of $s_j^2$ are inflated by differences in $f_0(\bX_i)$ driven by variation in the $\bX_i$ within each cluster.
In other words, if the control units are widely dispersed and the expected outcome $f_0(X_i)$ changes rapidly as $\bX_i$ changes, then the variation in $f_0(X_i)$ will be absorbed by the $s_j^2$ terms, resulting in overly large standard errors.
This phenomenon was observed in our simulations, as discussed in Section~\ref{sec:Inference-Results}.
This inflation is not necessarily bad, however, as we would expect it to grow roughly in proportion to the size of the bias term, which is also driven by these differences (see the $f_0(X_t) - f_0(X_i)$ terms in Equation~\ref{eq:bias-variance}).
One can thus view these SEs as predicting overall error (see, e.g., the discussion of standard errors as predicting overall error in \citet{sundberg2003conditional} and discussion of expanding of standard errors to include bias in \citet{weidmann2021missing}).

Our inference goal is the Root Mean Squared Error (RMSE) of our estimate, conditioned on the treatment group's covariates, represented by $\sqrt{E[(\hat \tau - \tau)^2 | \mathbf{X}_\mathcal{T}]}$.
Other options are possible, as is well illustrated by the literature.
\cite{kallus2020generalized} focuses on the conditional standard error, essentially the variance term.
Abadie and Imbens, along with \cite{otsu2017bootstrap}, target the unconditional coverage.
\cite{ferman2021matching} aims at the unconditional type-I error rate.
In particular, it is important to note, as \cite{imbens2015causal} mentions, that the conditional and unconditional standard errors differ.

\section{Simulation Studies}
\label{sec:simulation}

To understand how CSM performs, we consider a range of simulation studies.
First, we examine a simple simulation based on the toy examples in Section \ref{sec:toy}, where local matches tend to be important.
We then consider a few canonical simulated datasets taken from the literature to assess general performance.
Finally, we assess our inferential method.

\subsection{Methods}
\label{sec:sim_methods}

We compare CSM to a variety of popular matching, balancing, outcome regression, propensity score, and doubly robust methods, as described in Table \ref{tab:sim_methods}.\footnote{We initially also included CEM using synthetic controls within each cell and caliper matching using simple averages within each caliper, but, for clarity, we now exclude these results since their performances typically lie between the performances of CSM and CEM.}

\begin{table}[t]
\begin{tabular}{|l|l|l|}
\hline
\rowcolor[HTML]{C0C0C0} 
Method class                         & Method name & Description  \\ \hline
Baseline                             & diff & Difference-in-means estimator \\ \hline
                                     & match-1NN   & One nearest neighbor matching \\ \cline{2-3} 
\multirow{-2}{*}{Matching}           & match-CEM   & Coarsened exact matching \\ \hline
                                     & or-lm       & Linear model on all two-way interactions  \\ \cline{2-3} 
\multirow{-2}{*}{Outcome model} & or-BART     & Bayesian additive regression tree (BART) \\ \hline
                                     & ps-lm       & \begin{tabular}[c]{@{}l@{}}Logistic model for propensity score \\ 
                                     on all two-way interactions\end{tabular} \\ \cline{2-3} 
\multirow{-2}{*}{Propensity score}   & ps-BART     & BART with binary outcome (probit link) \\ \hline
                                     & bal-SBW1    & \begin{tabular}[c]{@{}l@{}}Stable balancing weights (SBW) \\ 
                                     on marginal means \end{tabular} \\ \cline{2-3} 
\multirow{-2}{*}{Balancing}          & bal-SBW2    & \begin{tabular}[c]{@{}l@{}}SBW on marginal and two-way \\ 
                                     interaction means \end{tabular} \\ \hline
                                     & dr-AIPW1    & \begin{tabular}[c]{@{}l@{}}Augmented inverse propensity weighting \\ 
                                     (AIPW) using SuperLearner on \\ 
                                     linear models\end{tabular} \\ \cline{2-3} 
                                     & dr-AIPW2    & \begin{tabular}[c]{@{}l@{}}AIPW using SuperLearner on \\ 
                                     machine-learning models \end{tabular} \\ \cline{2-3} 
                                     & dr-TMLE1    & \begin{tabular}[c]{@{}l@{}}Targeted maximum likelihood estimation \\ 
                                     (TMLE) using SuperLearner on \\
                                     linear models\end{tabular} \\ \cline{2-3} 
\multirow{-4}{*}{Doubly robust}      & dr-TMLE2    & \begin{tabular}[c]{@{}l@{}}TMLE using SuperLearner on \\ 
                                     machine-learning models \end{tabular} \\ \hline
\end{tabular}
\caption{Methods used in simulation studies. References for methods are: CEM \citep{iacus2012causal}, BART \citep{chipman2009bart}, SBW \citep{zubizarreta2015stable}, AIPW \citep{robins1994estimation}, SuperLearner \citep{van2007super}, and TMLE \citep{van2006targeted}.}
\label{tab:sim_methods}
\end{table}

We use default settings for all of the algorithms to standardize comparisons.
We implement BART, TMLE, and AIPW in \texttt{R} using defaults from the \texttt{dbarts} \citep{dorie2023dbarts}, \texttt{tmle} \citep{gruber2012tmle}, and \texttt{AIPW} \citep{zhong2021aipw} packages, respectively.
For SuperLearner, the linear models include a simple mean, linear regression, and generalized linear regression models; the machine-learning models include the linear models as well as generalized additive models, random forests, BART, and XGBoost \citep{chen2016xgboost}.

To standardize comparisons for the matching methods, for each dataset we use the distance metric implied by the covariatewise caliper generated by coarsening each numeric covariate into five equally spaced bins.
I.e., we use $d_V^{(\infty)}$ and a diagonal $V$ with entries $V_{kk} = 5 / (max(X_k) - min(X_k))$.
We do not tune the covariatewise caliper, assuming zero domain knowledge about variable importance.
We conduct CEM using the same uniform coarsening of each covariate into five bins.

\subsection{Toy example simulation}
\label{sec:Toy-example-simulation}
To build intuition for situations in which CSM may perform well, we first show results from the simulation based on the toy examples in Section \ref{sec:toy}.
We simulate the covariates for 50 treated units each from multivariate normal distributions centered at $(0.25, 0.25)$ and $(0.75, 0.75)$, and 225 control units each from multivariate normal distributions centered at $(0.25, 0.75)$ and $(0.25, 0.75)$, all with covariance matrices $\begin{bmatrix} 0.1^2 & 0 \\ 0 & 0.1^2 \end{bmatrix}$.
We then add 100 control units distributed uniformly on the unit square to ensure we do have some overlap.
We finally generate outcomes for each unit $(x_1, x_2)$ as:
$$Y = f_0(x_1, x_2) + Z \cdot \tau(x_1, x_2) + \epsilon,$$
where $f_0$ is the density function for a multivariate normal distribution centered at $(0.5, 0.5)$ with covariance matrix $\begin{bmatrix} 1 & 0.8 \\ 0.8 & 1 \end{bmatrix}$, 
$\tau(x_1, x_2) = 3x_1 + 3x_2$ is the true treatment effect, 
and $\epsilon \sim N(0, 0.5^2)$ is homoskedastic noise.
Figure \ref{fig:sim_toy_ex} plots a sample simulated dataset.
\begin{figure}[t]
    \centering
    \includegraphics[width=0.7\textwidth]{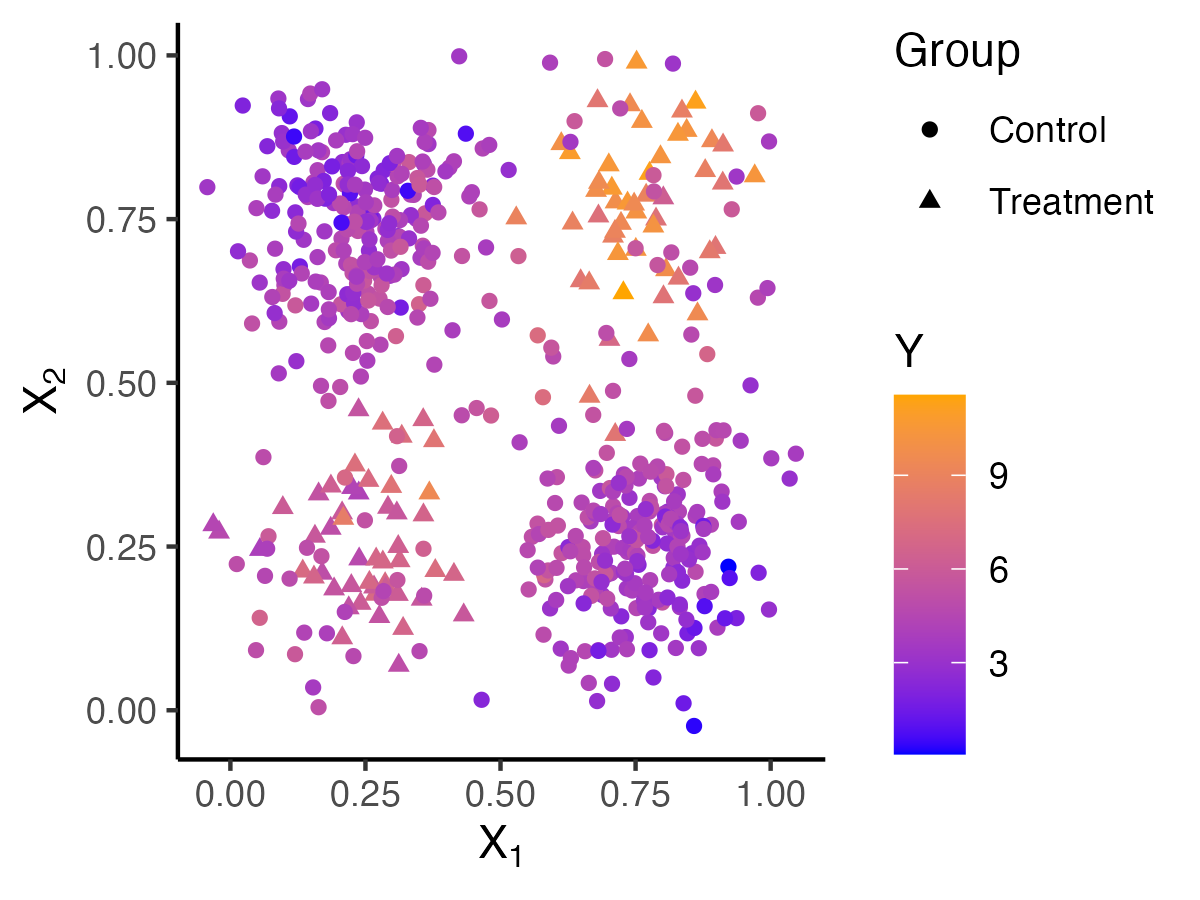}
    \caption{Sample simulated toy dataset.}
    \label{fig:sim_toy_ex}
\end{figure}

Figure \ref{fig:sim_toy_results} shows the root mean-squared error and absolute bias of the point ATT estimates for the various methods.
\begin{figure}[t]
    \centering
    \includegraphics[width=0.7\textwidth]{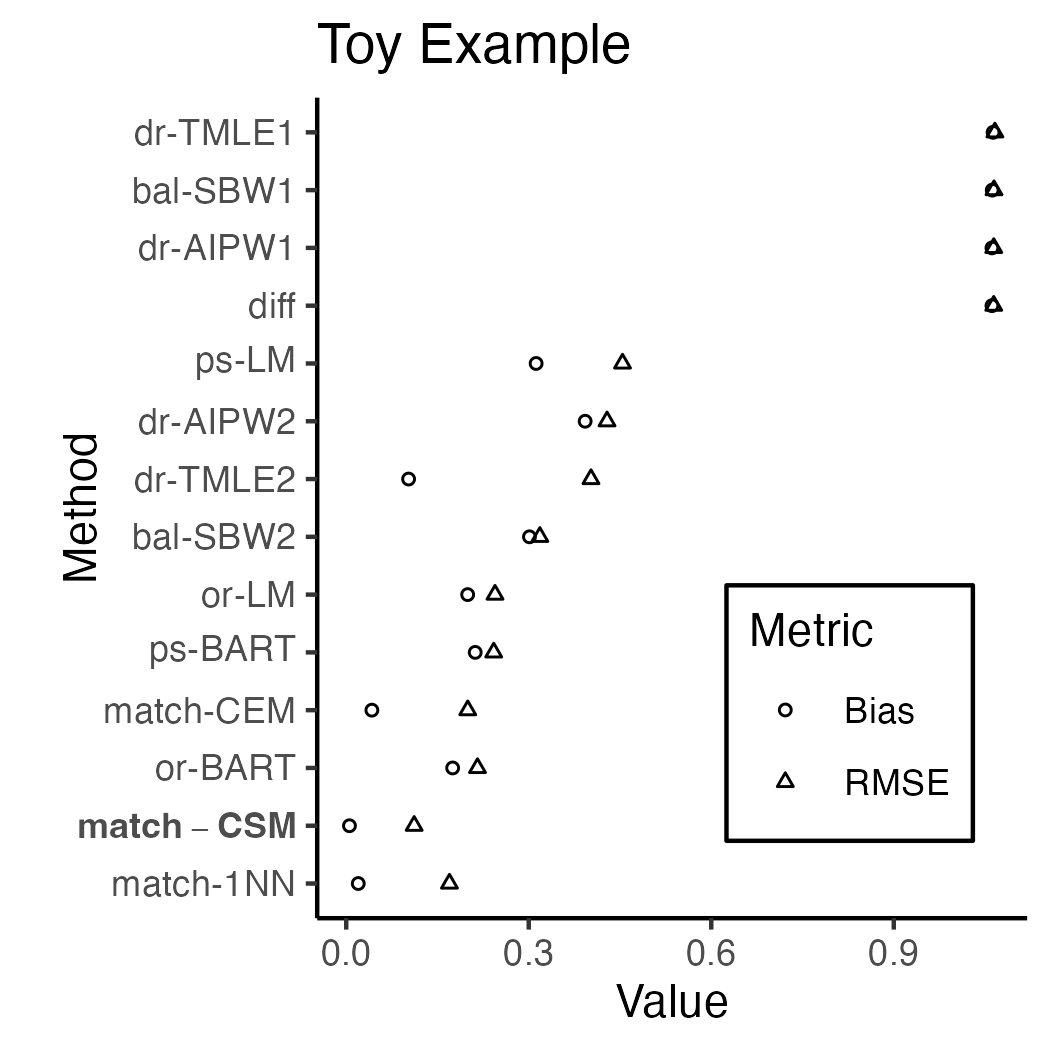}
    \caption{Method results on simulation based on toy example across 250 simulations.}
    \label{fig:sim_toy_results}
\end{figure}
We see that CSM performs well, even compared to complex machine-learning methods.

We emphasize that the goal of this simulation is not to demonstrate CSM's performance, but rather to illustrate settings in which we may expect good performance.
In the toy example, the interaction between the two covariates drives both the control potential outcome function $f_0(X_1, X_2)$ and the implicit propensity score.
Because the covariates are interacted with each other, joint covariate balance is very important.
As a result, those methods that do not target joint balance or otherwise do not have access to the interaction term between $X_1$ and $X_2$ perform poorly due to high levels of bias.

\subsection{Canonical simulations}
\label{sec:Canonical-simulations}
The datasets from \citet{kang2007demystifying}, \citet{hainmueller2012entropy}, and the 2016 American Causal Inference Conference competition \citep{dorie2019automated} have canonically been used to compare methods for observational causal inference.
Table \ref{tab:simdata} contains basic information about the settings we used for each dataset.
Please see the original papers for further details.

\begin{table}[t]
\begin{tabular}{|l|l|l|l|l|}
\hline
\rowcolor[HTML]{C0C0C0} 
Dataset                   & $p$ & $n_t$         & $n_c$         & \begin{tabular}[c]{@{}l@{}}Heterogeneous \\ effects?\end{tabular} \\ \hline
Kang \& Schafer, (2007)   & 4             & $\approx 500$   & $\approx 500$   & No                                                                          \\ \hline
Hainmueller et al. (2012) & 6             & 50            & 250           & No                                                                          \\ \hline
ACIC 2016                 & 10            & $\approx 350$ & $\approx 650$ & Yes                                                                         \\ \hline
\end{tabular}
\caption{Basic descriptions of simulated datasets. 
    We use the \citet{kang2007demystifying} as-is from the original paper; the ``high overlap'' and ``highly nonlinear'' outcome model condition for the data from \citet{hainmueller2012entropy}; and the ``step-function'' propensity and treatment model conditions for the ACIC 2016 data using 10 numeric covariates.}
\label{tab:simdata}
\end{table}

Figure \ref{fig:sim_results} summarizes the results of our simulations using these datasets and illustrates the tradeoffs made by CSM.
\begin{figure}[t]
    \centering
    \includegraphics[width=\textwidth]{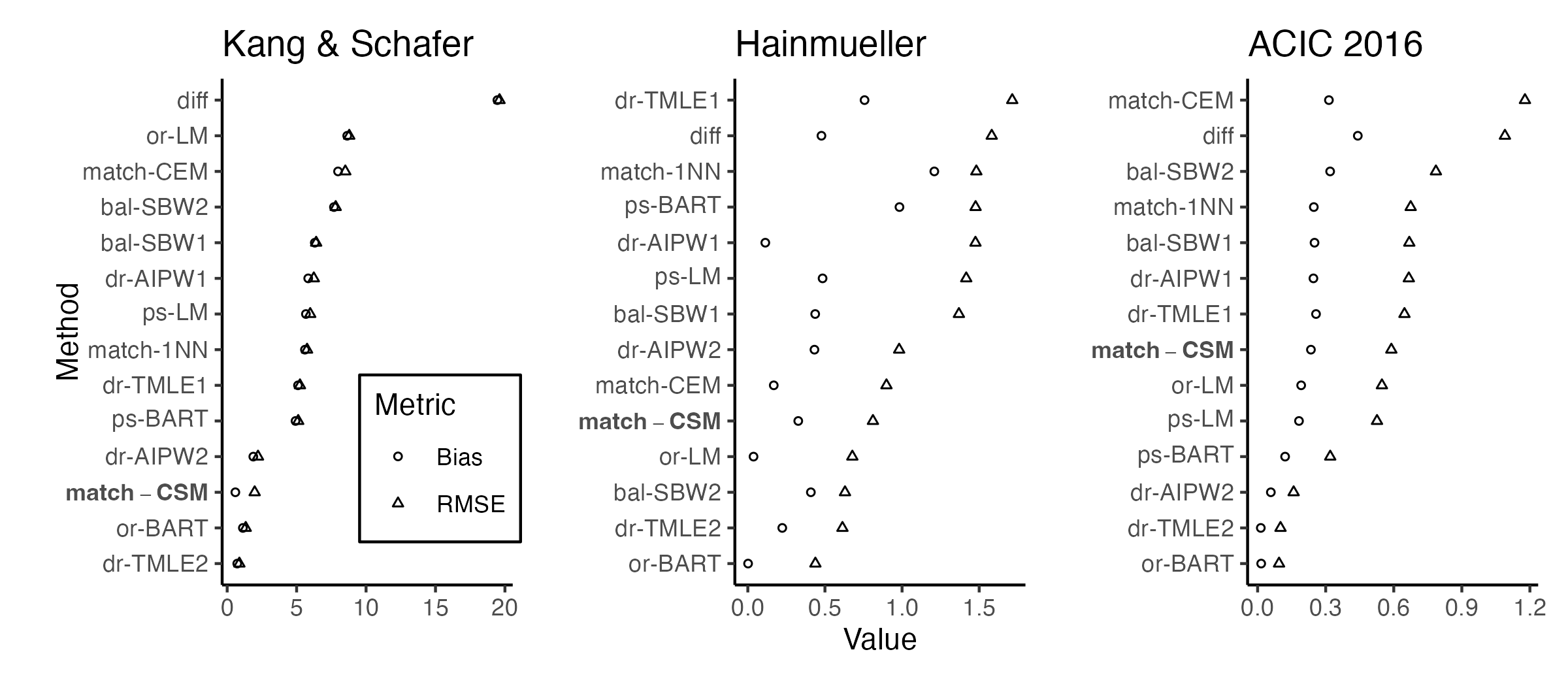}
    \caption{Method results on canonical simulated datasets across 250 simulations. 
        Results for CSM highlighted in bold font.  Four methods truncated at 1 for displaying the results.}
    \label{fig:sim_results}
\end{figure}
We see that CSM tends to outperform the simple matching, balancing, and outcome or propensity-score modeling approaches, but that it is outperformed by the complex machine-learning approaches known to do well on these types of datasets \citep{dorie2019automated}.
CSM is designed to maximize performance while preserving the simple intuitions underlying exact matching.
Under such constraints, it is remarkable that CSM performs nearly as well as black-box machine-learning approaches on the \citet{kang2007demystifying} and \citet{hainmueller2012entropy} datasets.

The trends in Figure \ref{fig:sim_results} highlight important concepts regarding locality in observational causal inference.
As the number of covariates increases, the matching methods tend to deteriorate in performance relative to the complex modeling methods.
Caliper-based matching methods only use local information.
This behavior protects them against incorrect extrapolation, but also prevents them from correctly extrapolating in cases where overlap is unnecessary.

In many simulated (and real) datasets, extrapolating marginal effects can improve counterfactual predictions.
For example, suppose increasing $X_1$ clearly increases outcomes among controls which have low values of $X_2$, and we would like to predict counterfactual outcomes for treated units with high values of $X_2$.
A linear model fit to the controls may extrapolate and assume that increasing $X_1$ always increases outcomes regardless of the value of $X_2$, while an exact matching estimator would not do so.
Increasing the number of covariates exaggerates this behavior, since models extrapolate more marginal trends.
Figure \ref{fig:sim_results} shows us that this type of model-based marginal extrapolation improves results for these simulated datasets, since in the considered DGPs there are few high-dimensional interactions to falsify such extrapolation.

Matching aims to produce transparent, model-free causal inferences.
As shown by Figure \ref{fig:sim_results}, this transparency generally comes with a cost in terms of performance.
The simulations show that while avoiding model-based extrapolation controls bias in the worst case (e.g., Proposition \ref{prop:biasbd_lip}), it can be too conservative in settings where such extrapolation is helpful.
Whether these costs are outweighed by the ability to clearly explain how counterfactual predictions are made depends on the particular causal inference setting.
Nonetheless, CSM performs competitively in these canonical simulations, making it an attractive option in settings where explainability is paramount.

\subsection{Assess the Quality of the Variance Estimator}
\label{sec:Inference-Results}

\begin{figure}[t]
    \centering
    \includegraphics[width=\linewidth]{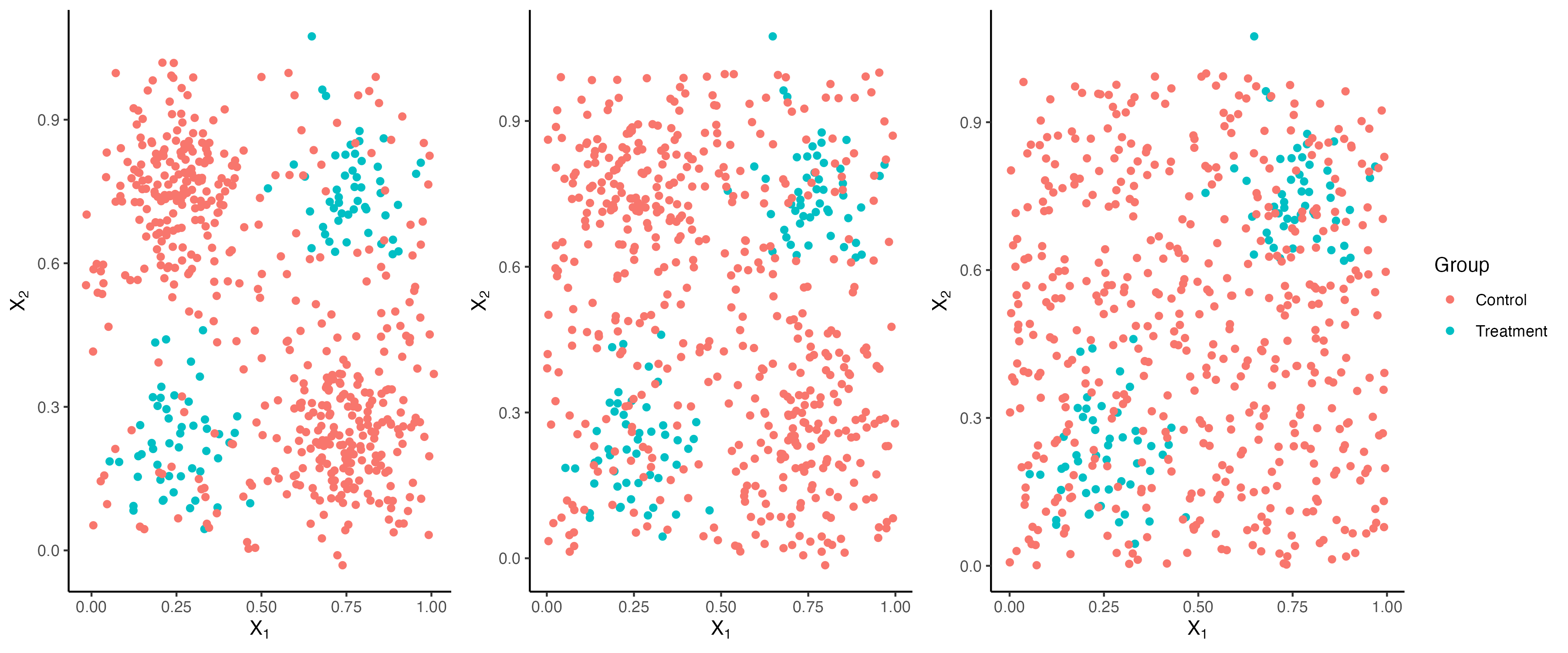}
    \caption{Toy example in different degree of overlaps. From left to right: very low, medium, and very high}
    \label{fig:sim_toy_3_overlaps}
\end{figure}

We next assess the variance estimator described in Section~\ref{sec:var-est}.
We revisit the example from Section~\ref{sec:Toy-example-simulation} but vary the degree of overlap.
Specifically, we increase the overlap by increasing the proportion of control units distributed uniformly on the unit square and reducing the proportion of controls centered at (0.25, 0.75) and (0.75, 0.25).
For the range of overlap, see the three illustrative datasets in Figure~\ref{fig:sim_toy_3_overlaps}.
These scenarios correspond to the same distribution of treated units, but different pools of potential controls to select matches from.
For each degree of overlap, we generate 500 trials.

Our primary estimation target is the SATT.

To capture tuning $c$ for bias reduction, we use an adaptive strategy for the caliper size, shrinking the caliper to select no more than five units in regions with great density, and expanding it to include at least one unit in regions of low density.
With larger calipers, the synthetic control optimization tends to select only a few units of the set when we have a low-dimensional covariate set, and the selected units tend to be more distant from the treated unit than others within the neighborhood defined by the caliper; for future work we could follow \citet{ben2021augmented} and regularize the weights to take more advantage of rich areas of control units.

\begin{table}[ht]
\centering
\centering
\begin{tabular}{lrrrrrr}
  \hline
Overlap & $\tilde{N}_C$ & avg $\widehat{SE}$ & avg $SE$ & Bias & RMSE & Coverage \\ 
  \hline
Very Low & 23.9 & 0.12 & 0.11 & 0.04 & 0.12 & 0.95 \\ 
  Low & 39.2 & 0.10 & 0.10 & 0.02 & 0.10 & 0.94 \\ 
  Medium & 51.5 & 0.09 & 0.08 & 0.01 & 0.08 & 0.96 \\ 
  High & 61.3 & 0.08 & 0.08 & 0.01 & 0.08 & 0.95 \\ 
  Very High & 70.2 & 0.08 & 0.08 & 0.01 & 0.08 & 0.94 \\ 
   \hline
\end{tabular}
\caption{Estimated vs. actual standard error for varying degrees of overlap. Columns are average effective sample size of control group (vs. 100 treated units), average estimated standard error, average true standard error, Bias, RMSE, and coverage of the SATT. Monte Carlo SEs are around 0.025 for the true standard error and RMSE, 0.01 for coverage, and $<0.001$ for the other measures.}
\label{tab:toy_overlap_A-E}
\end{table}

Table \ref{tab:toy_overlap_A-E} presents the results.
 The column ``avg $\widehat{SE}$'' shows the estimated standard error averaged over 500 iterations.
The ``avg $SE$'' column shows the true conditional standard error,\footnote{This value is   calculated as the standard deviation of $\hat{\tau} - bias - \tau$ across the simulations in order to remove variation due to simulation iterations having different true SATT values.}
Overall, we observe that our method slightly overestimates the true standard error due to the inflation caused by absorbing variation in $f_0$, as discussed in Section~\ref{sec:var-est}.

As expected, both the true standard error and bias go down as overlap increases, due to the steadily increasing effective sample size and  improved ability to find close matches.
In particular, when there is low overlap, we re-use many control units, and only at high overlap is our control group getting close to the treatment group in size.
Additionally, with high overlap, there are more control units near the concentration of treated units, allowing for closer matches and thus less bias.

Finally, the ``Coverage'' column shows the coverage of the 95\% confidence interval. The CSM method effectively maintains a small degree of bias and gives reasonably estimated standard errors, resulting in CI coverage close to the nominal 95\% across all the scenarios.

An alternate form of inference would some form of bootstrap.
Unfortunately, we found that bootstrap methods, including the weighted bootstrap in \cite{otsu2017bootstrap} applied to our context, and the na\"ive casewise bootstrap, did not work effectively.
We leave further development of bootstrap inference to future work.

\section{Applied Example: Effect of Improving Education Quality in Schools in Brazil}
\label{sec:ferman-data-analysis}

We illustrate our approach via a within study comparison based on a randomized control trial, ``Jovem de Futuro'' (Young of the Future), of schools in Brazil, following a study by \cite{barros2012impacto}.
Data from this intervention, targeting education quality, was also analyzed as a within study comparison by \citet{ferman2021matching}.

In a within study comparison, we try to recover a known impact via our observational study method.
Within study comparisons allow for assessment of how observational methods serve ``in the wild.''
A typical early example of this is the famous LaLonde paper \citep{lalonde1986evaluating}, but for a more extended discussion, see \cite{cook2008three}.
In particular, we take units enrolled in an experimental trial, but who did not receive treatment (the experimental control group) as our ``treatment'' and match them to a larger pool of units not in the experiment at all; see \cite{weidmann2021lurking} for further discussion of this form of within study comparison.

Jovem de Futuro combined two efforts: a) offering strategies and tools for school management boards to improve efficiency; b) offering grant money for improving education quality, conditional on school's standardized test score passing a certain threshold. 
The intervention lasted three years, and there were three rounds delivered: 2010-2012 (the 2010 implementation), 2011-2013 (the 2011 implementation), and 2012-2014 (the 2012 implementation).
We focus on the 2010 implementation.

The 2010 implementation had 15 treated schools and 15 control schools in Rio de Janeiro, and 39 treated schools and 39 control schools in Sao Paulo.
We set $Z_i = 1$ for the control schools who applied for the intervention program but were assigned to the no-intervention group, and $Z_i = 0$ for the non-participating schools.
Since there is no actual intervention for either the $Z_i=1$ and $Z_i=0$ schools, we expect the ``treatment effect'' to be zero, if our methods succeed in removing any confounding selection bias.
The sample sizes for Rio de Janeiro and for Sao Paulo are $n_T^{Rio} = 15, n_C^{Rio} = 966$ and $n_T^{SP} = 39, n_C^{SP} = 3481$. 

For each school, we have four matching variables: $X_1, X_2$, and $X_3$ are standardized test scores from 2007 to 2009, and $X_4$ is  an indicator of being in Sao Paolo.
Table \ref{tab:marginal_means} shows the pre-matching covariate mean difference. Overall, the schools in the RCT scored lower than those not in the RCT. 

Our outcome is the test scores in 2010(the outcome following the first year of intervention).

\begin{table}[ht]
    \centering
    \begin{adjustbox}{width=\textwidth}
    \begin{tabular}{lccc}
        Variable & Control (Z = 0) & Treated (Z = 1) & Difference (treated - control) \\
        \hline
        Score 2007 ($X_1$) & 4.7 & -2.8 & -7.5 \\
        Score 2008 ($X_2$) & 0.83 & -0.99 & -1.8 \\
        Score 2009 ($X_3$) & 2.3 & -2.5 & -4.8 \\
        In Sao Paulo ($X_4$) & 78\% & 72\% & -6\% \\
    \end{tabular}
    \end{adjustbox}
    \caption{Marginal means of the covariates in \cite{barros2012impacto} before matching. All variables have been multiplied by 100.}
    \label{tab:marginal_means}
\end{table}

\subsection{CSM Procedure}
We start our CSM analysis by choosing $L^\infty$ as the distance metric.
Next, we set the covariate-wise calipers.
For the standardized test scores $X_1$ to $X_3$, we set the covariate-wise calipers to 0.2, meaning we want matched schools to differ at most 0.2 standard deviations in test scores.
For $X_4$ (indicator of being in Sao Paolo), we set the covariate-wise caliper to $1/1000$, a small value to ensure the matching of $X_4$ is exact.  

The distance-metric caliper $c$ is usually set to a default of 1, as done in Proposition \ref{prop:CSM-is-MIB}, to give direct interpretation to the covariate-wise calipers $\boldsymbol{\pi}$.
In practice, however, given $\boldsymbol{\pi}$ we can tune $c$ to guarantee each treated unit gets enough but not too many matched controls, thereby optimally trading off data use and potential bias (see Section \ref{sec:assess-bias-variance-scm}).
When we set $c < 1$, we obtain tighter matches than initially planned with $\boldsymbol{\pi}$.
When we set $c > 1$, our matches can be worse than initially planned.

To tune $c$, we plot histograms of the distances of the top-1, top-2, and top-3 matches for each treated unit, and look for a natural break.
We see that, for most treated units, even the third worst match has a distance of around 0.35 or less.
We can thus focus on bias reduction, and set our distance-metric caliper $c$ to $0.35$.
See Appendix \ref{app:caliperchoice} for further details.

Some of our treated units have no controls within this set minimum caliper.
For these units, we use an adaptive caliper, matching the units to their nearest neighbor.
We then apply the synthetic control method to obtain weights for each set of matched controls.
Finally, we use these weighted controls to construct the point estimate and estimated standard error, using the inferential method outlined in Section \ref{sec:var-est}.

\subsection{Assessing the Covariate Balance}
Our chosen caliper results in 49 (91\%) treated units
having at least one matched control within a radius of $c = 0.35$.
These are our ``feasible'' treated units.
Among the 49 feasible treated units, all matched controls are within 0.2 standardized scores for the years 2007, 2008, and 2009. 
Figure \ref{fig:love-plot-ferman} illustrates how the differences in marginal means for three matching variables vary when adding infeasible treated units in order of increasing adaptive caliper size, until all 54 treated units are included.

By Proposition 4.3, the marginal mean difference of each covariate must be within $c\cdot\pi_k$, where $\pi_k$ is the covariate caliper of the $k$-th covariate, and $c$ is the distance-metric caliper.
For $X_1$ to $X_3$, this value is $0.35 \cdot 0.2 = 0.07$.
For $X_4$, this value is $0.35 \cdot \frac{1}{1000} = 3.5 \times 10^{-4}$.
The leftmost point in each plot in Figure \ref{fig:love-plot-ferman} shows the difference in marginal means between the treated units and their synthetic controls for the 49 feasible treated units, with differences all well below the corresponding $c\cdot \pi_k$ values. This result is not surprising: $c\cdot \pi_k$ represents a worst-case guarantee, and imbalances within each matched control set can cancel out. 

An important use of Figure \ref{fig:love-plot-ferman} is to provide a diagnostic on covariate balance.
The plot assesses the approximate joint balance through marginal covariate balance while including increasingly infeasible units.
While joint and marginal balance is guaranteed for feasible units, it is not guaranteed when more infeasible units are added. In this dataset, the marginal balances remain good even when all infeasible treated units are included.

\begin{figure}[t]
    \centering
    \includegraphics[width=0.9\linewidth]{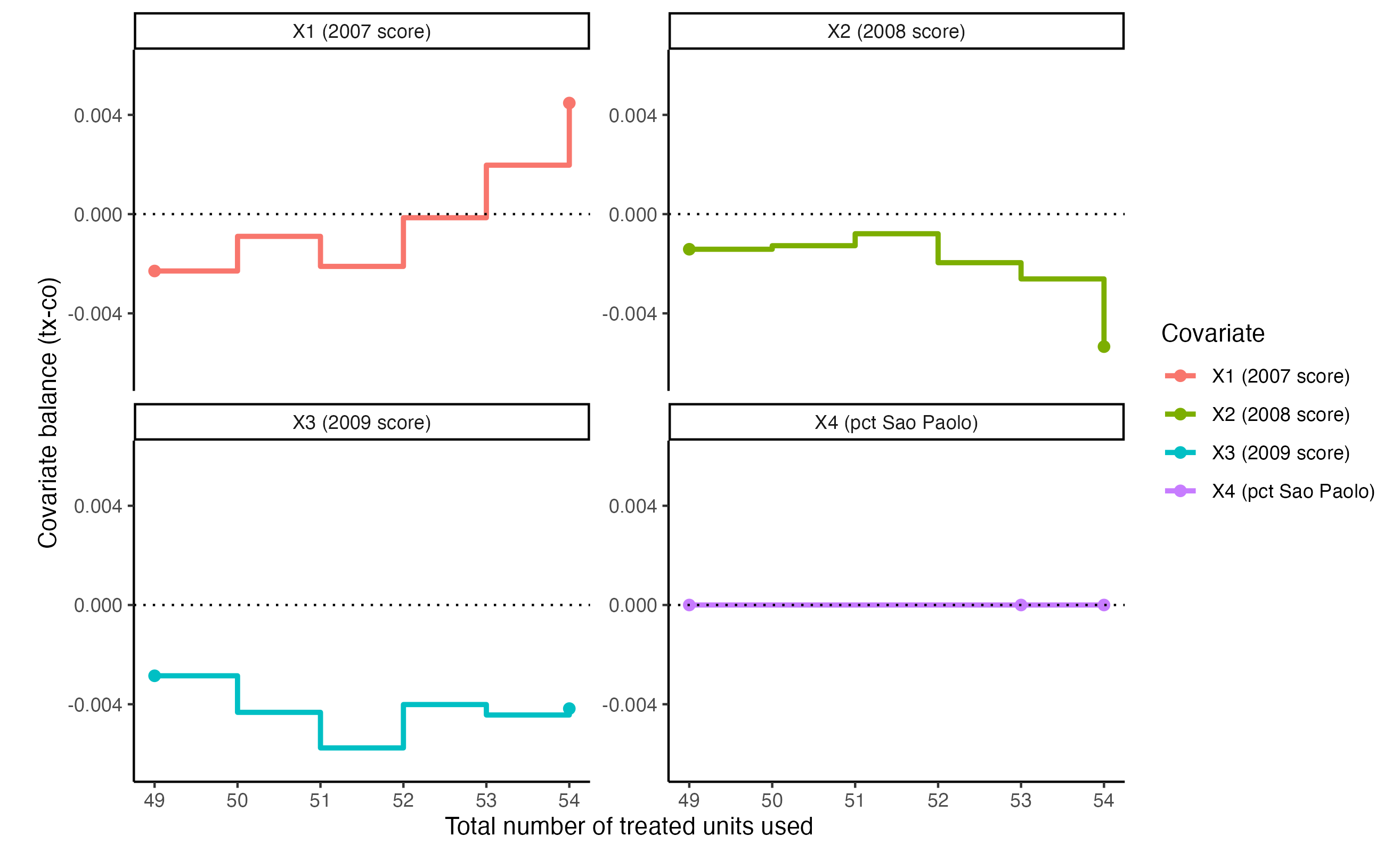}
    \caption{Love plot on all four matching variables as increasingly infeasible units are included.}
    \label{fig:love-plot-ferman}
\end{figure}

\subsection{Assessing the Impact of Synthetic Controls}
\label{sec:assess-bias-variance-scm}
 The synthetic control step after matching ideally improves the comparability of the controls matched to each treated unit, but can also reduce the effective sample size of the control group by downweighting some units.
 We next assess this bias-variance trade-off, comparing SCM with classic radial matching (where we take the simple average of all close controls as the counterfactual for each treated unit in turn) and one-nearest-neighbor (1-NN) matching.
 Both can be viewed as weights on the controls:

\begin{enumerate}
    \item Simple average weight: let $w_{jt} = \frac{1}{|\Ct|}$ for each control unit. This is the default weighting method used in the matching literature after matches are made.
    \item One-nearest-neighbor\footnote{If unit $t$ has $k$ unites tied for nearest neighbor, we assign each a weight of $w_{jt} = \frac{1}{k}$.}: 
        let $w_{jt} = \begin{cases}
            1 & \text{if unit } j \text{ is closest to unit } t \\
            0 & \text{otherwise.}
        \end{cases}$ 
        
        This weight corresponds to 1-to-1 matching.
\end{enumerate}

The results for SCM, average, and 1-NN weights on the feasible treated units are shown on Table~\ref{tab:synthetic_controls}.
We analyze the bias-variance trade-off through two metrics: the mean covariate difference (a proxy for bias, Equation \ref{eq:bias-variance}) and the effective sample size (ESS, a proxy for variance, Equation \ref{eq:true-variance}). 
The mean covariate difference, shown in the first column of Table~\ref{tab:synthetic_controls}, is the first-order term in Taylor’s expansion of bias.
It is also referred to as the extrapolation bias in \cite{kellogg2021combining}, representing how far the weighted control is from the treated unit in $d_V^{\infty}$ in the covariate space.
SCM performs better than simple average and 1-NN since it is optimized to minimize this bias.

Under homoskedasticity, the variance of an estimator is inversely associated with ESS, meaning a larger ESS implies smaller variance.
ESS will be larger when more controls are used and when weights are more evenly distributed among those selected.
Average weights achieve the highest ESS by assigning non-zero weights to all 202 distinct controls within the caliper distance of at least one of the 49 feasible units. In contrast, 1-NN uses only 49 distinct controls, as each treated unit is matched to a unique nearest neighbor. 
SCM falls between these scenarios, assigning zero weights to some control unit otherwise within caliper distance, but using multiple controls where possible.
Overall, SCM performs well in the bias-variance trade-off by achieving the lowest imbalance (bias) and moderate variance.

\begin{table}[ht]
    \centering
    \begin{adjustbox}{width=\textwidth}
    \begin{tabular}{lcccc}
               & Mean Individual & Median Individual &   &Number of \\
        Method &  Imbalance & Imbalance & ESS  &  Unique Controls \\
        \hline
        Feasible Treated Group & - & - & 49 & - \\
        CSM     & 0.057 & 0.000 & 95 & 155 \\
        Average & 0.118 & 0.091 & 202 & 580 \\
        1-NN    & 0.158 & 0.148 & 49 & 49 \\
    \end{tabular}
    \end{adjustbox}
    \caption{Effect of synthetic controls on bias-variance trade-off metrics of the FSATT estimate, comparing to average and 1-NN weighting.}
    \label{tab:synthetic_controls}
\end{table}

\subsection{FSATT and SATT Estimates}
\label{sec:FSATT}
We finally present the estimated treatment effect on the feasible units and also show the influence of infeasible treated units on the ATT estimate.
Before interpreting the results, recall that, given our within study design, we expect to see no treatment impact.
In other words, assuming our estimation is sound, we should expect an ``impact estimate'' of zero.

Figure \ref{fig:fsatt-ferman} presents the main results of the analysis, providing a series of confidence intervals given a steady increase in the number of units kept in the analysis.
Each confidence interval is constructed using the method in Section \ref{sec:var-est}.
The leftmost confidence interval, representing the feasible set, covers zero, indicating no significant effect as desired. As we include more infeasible units, the ATT varies slightly, but is overall stable, and all confidence intervals include 0.
This stability aligns with Figure \ref{fig:love-plot-ferman}, which shows covariate balances only vary slightly as more infeasible units are included.

\begin{figure}[t]
    \centering
    \includegraphics[width=\linewidth]{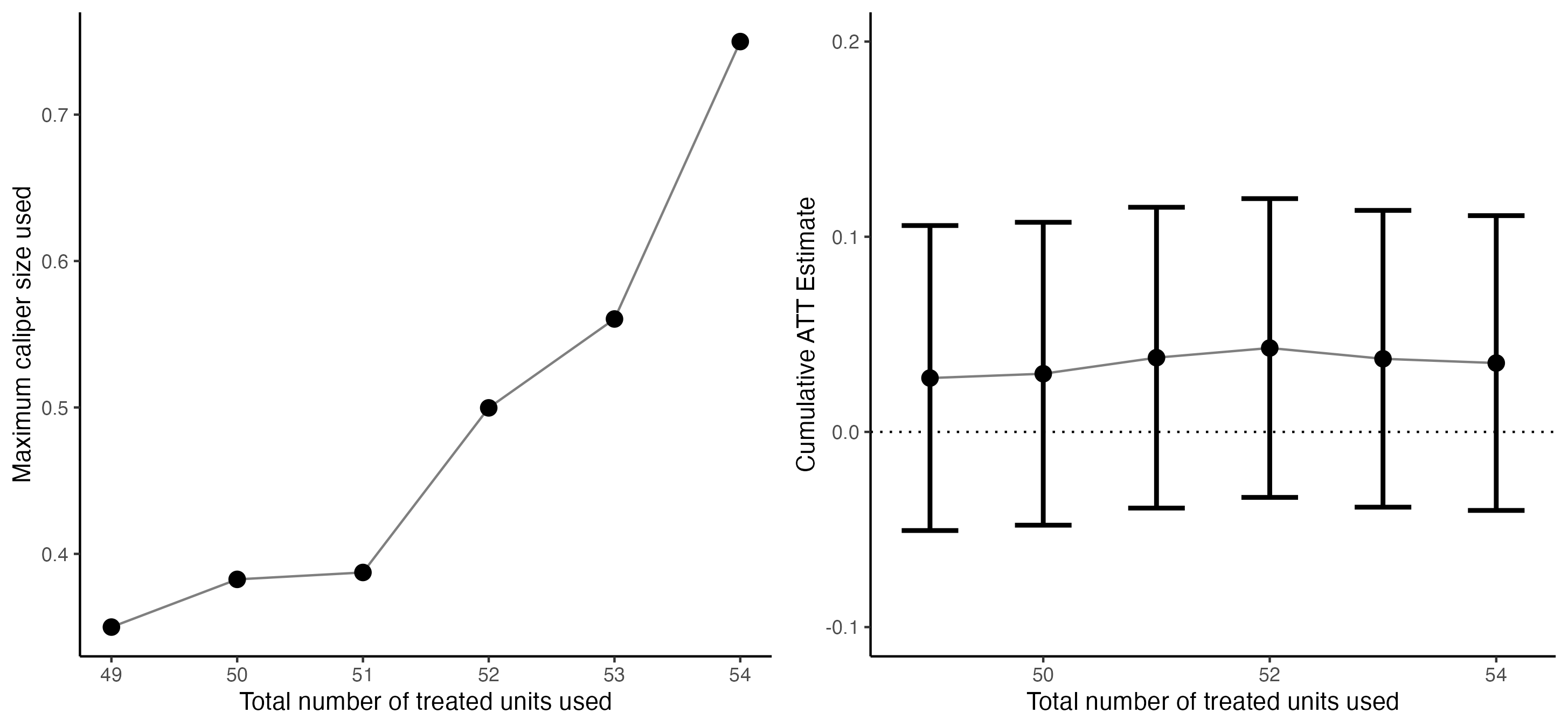}
    \caption{Estimate-estimand tradeoff plot. Maximum adaptive caliper sizes range from 0.35 (for the FSATT) to 0.75 (for the full ATT).}
    \label{fig:fsatt-ferman}
\end{figure}

The SATT estimates from CSM and the other two methods are presented in Table~\ref{tab:point_estimates}. Using the point estimates and standard errors, all methods conclude no treatment effect. The standard error follows a trend where 1-NN is the largest, Average is the smallest, and CSM is in the middle, consistent with the trend we observed in Section~\ref{sec:assess-bias-variance-scm}: CSM uses more data than 1-NN, but sacrifices some precision to ideally reduce bias.

\begin{table}[ht]
    \centering
    \begin{tabular}{l|cc}
        Method & Point Est. & SE \\
        \hline
        CSM & 0.035 & 0.039 \\
        Average & 0.013 & 0.038 \\
        1-NN & 0.008 & 0.053 \\
    \end{tabular}
    \caption{Point estimates and standard errors for all three methods.}
    \label{tab:point_estimates}
\end{table}

\section{Discussion}

Matching methods enable researchers to simply and intuitively draw causal conclusions from observational data.
In practice, however, standard matching approaches typically need to sacrifice either performance or transparency to achieve results comparable to modern optimization and machine-learning methods for causal inference.
To address this challenge, this paper introduces Caliper Synthetic Matching.
CSM builds on the spirit of exact matching by preserving four key principles: intuitive distances, local matches, ability to monitor matching success, and transparent estimates.
We find that CSM can achieve comparable performance to modern methods on standard benchmark datasets, particularly when many covariates exhibit interaction effects.

CSM provides a framework for transparent causal inference, and each stage may be extended in various ways.
Future work might reconsider the question of distance-metric and caliper selection, optimizing the distance-metric scaling matrix $V$ using held-out data, or incorporating estimated covariate densities to optimize caliper lengths (for work in this area, see \citet{parikh2022malts}).
Dimension reduction on the number of covariates, possibly by using prognostic scores (as done in, e.g., \citep{aikens2020pilot}), could also help mitigate the curse of dimensionality.
Interpretable estimation methods other than standard synthetic controls may also be used, such as the modified synthetic controls approaches proposed by \citet{abadie2021penalized} and \citet{ben2021augmented}.
Ultimately, CSM represents a promising addition to the toolkit for researchers conducting causal inference using matching methods.

\bibliographystyle{apalike}
\bibliography{refs.bib}

\newpage

\appendix

\section{Practical considerations}

\subsection{How to select a caliper}
\label{app:caliperchoice}

The choice of caliper size is a notorious practical problem for caliper-based matching methods.
Ideally, the researcher would have a covariatewise caliper $\boldsymbol{\pi}$ in mind, so that the distance-metric caliper can be simply to equal 1, as in Proposition \ref{prop:distmetriccal}.
More generally, given a fixed distance metric, the choice of $c$ should be chosen based on an a priori desired level of bias control rather than a post hoc assessment based on the observed data.
For example, a researcher who wants to ensure that all matches lie within 0.5 standard deviations of each other in each covariate, could restrict $d_V^{(\infty)}(\Xt, \Xj) \leq 0.5$ with $V = diag\{\frac{1}{sd(X_k)}, k=1, \dots, p\}$.

In practice, however, researchers may want to select a distance metric caliper $c$ that ``optimally'' trades off data use and potential bias.
To visualize this tradeoff, we suggest making a histogram of the distances between the treated and nearby control units in the data: Start by computing the $n_t \times n_c$ distance matrix used to identify the within-caliper control units for each treated unit. Then, plot the smallest $k$ distances $d(\mathbf{X}_t, \mathbf{X}_j)$ for each $t \in \mathcal{C}_t$ on a histogram to see how ``far" the control units tend to be from the treated units.
By doing this, we avoid needing to estimate the $p$-dimensional density of the treated and control units, which is generally challenging to do.

In Figure \ref{fig:hist-top-k-distances}, we created three histograms: closest distance, second closest distance, and third closest distance for each $t \in \Ct$ in the Brazilian school dataset discussed in Section \ref{sec:ferman-data-analysis}, for units that are exactly matched on the categorical covariates.
We see that $c=0.35$, a value slightly above the 75th percentile of the third closest distance, is a good cut-off since most $d(\Xt, \Xj)$ lies on the left of $0.35$.
The 0.35 cut-off also captures most of the third closest and second closest distances, indicating we would get at least three units matched for most of our treated units.
In general, one could look for peaks around particular values of $d(\Xt, \Xj)$, as expanding the fixed caliper to include those peaks can greatly increase the effective sample size of the data.

\begin{figure}[t]
    \centering
    \includegraphics[width=0.9\linewidth]{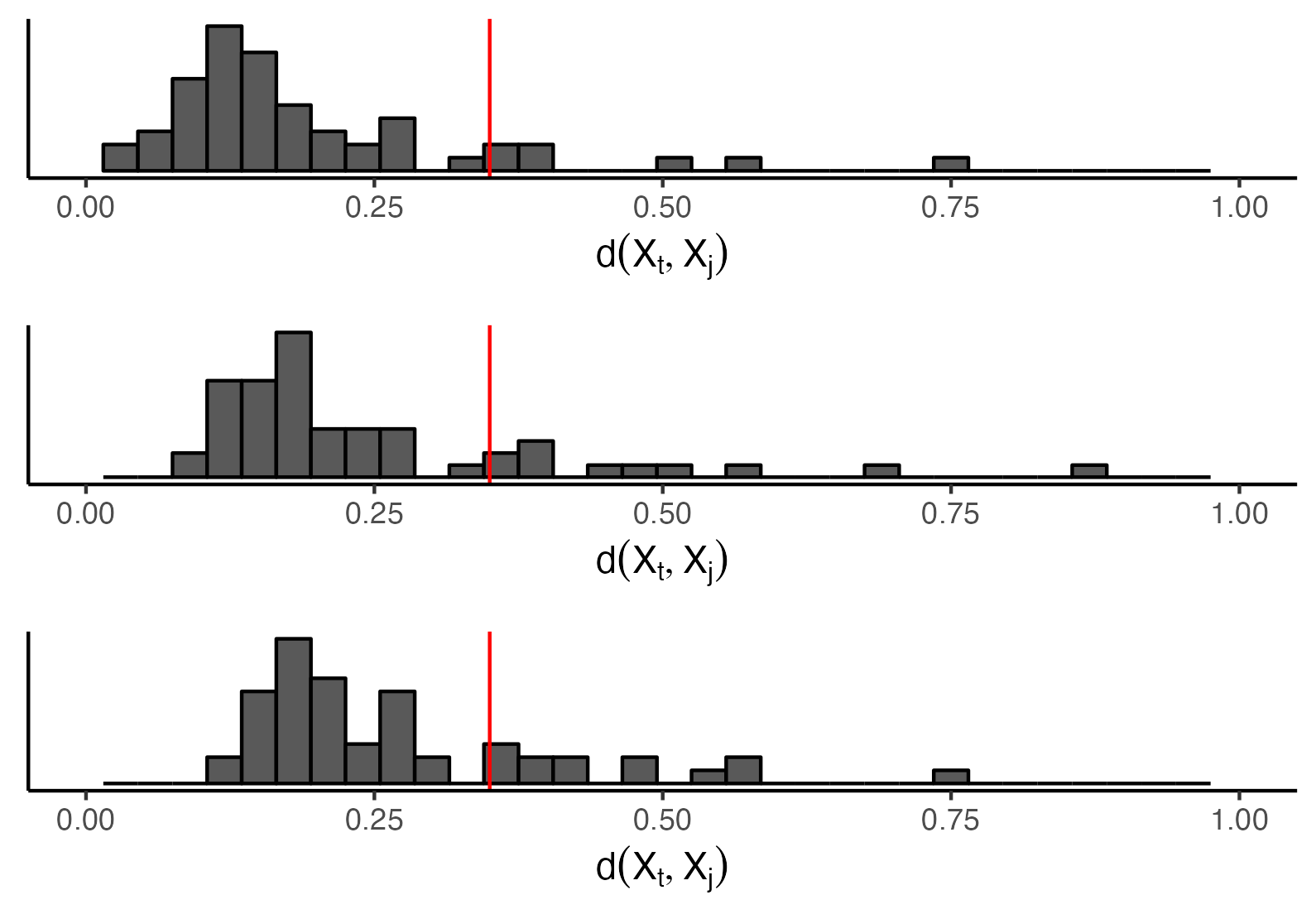}
    \caption{Histogram of distances between treated units and control units.}
    \label{fig:hist-top-k-distances}
\end{figure}

\section{Derivations and Technical details}

\subsection{Distance metrics in multiple dimensions}
\label{app:distmetrics}

Formally, a distance metric on $\Rp$ is a (non-negative) function $d(\cdot, \cdot): \Rp \times \Rp \to \R$ that satisfies the following for $x,y,z \in \Rp$:
\begin{enumerate}
    \item $d(x,y) = 0 \iff x = y$
    \item $d(x,y) = d(y,x)$
    \item $d(x,z) \leq d(x,y) + d(y,z)$ [Triangle inequality]
\end{enumerate}
These three properties satisfy our intuitive understanding of measuring distance between two points: we measure a distance of zero if and only if the two points concide with each other; the distance is the same regardless of which point we start measuring from; and the distance cannot be made shorter by passing through a third point.

Now consider a scaled $L_2$ norm on $\Rp$, which measures the length of a given vector $u \in \Rp$ as:
$$||u||_V^{(2)} = \sqrt{u^t (V^T V) u}.$$
We say that the scaled distance metric in Equation \ref{eq:l2dist} is \textit{induced} by the scaled $L_2$ norm defined above, since for $x,y \in \Rp$ we can write:
$$d_V^{(2)}(x,y) = ||x - y||_V^{(2)}.$$
The same relationship holds for the scaled distance metric in Equation \ref{eq:linfdist} and the scaled $L_2$ norm $||u||_V^{(\infty)} = \max|Vu|.$

Induced distance metrics in $\Rp$ possess desirable properties for $x,y \in \R$:
\begin{itemize}
    \item Translation invariance: for $c \in \R$, $d(x,y) = d(x+c, y+c)$
    \item Absolute homogeneity: for $c \in \R$, $d(cx, cy) = |c| d(x,y)$
\end{itemize}
Again, these properties are intuitive: moving two points by the same amount does not change the distance between them, and scaling two points scales the distance between them.

In deriving our formal results below, we will need the following simply lemma about our distance metrics:

\begin{lemma}
\label{lem:unitdist}
For any translation-invariant, absolutely homogeneous distance metric $d(\cdot, \cdot)$ on a metric space, $d(\mathbf a + c \mathbf v, \mathbf a) = c$ for $c \in \R$, $\mathbf a \in \Rp$ and unit vector $\mathbf v \in \Rp$.
\end{lemma}
\begin{proof}
\begin{align*}
    d(\mathbf a + c \mathbf v, \mathbf a)
    &= d(c \mathbf v, 0) &\text{[translation invariance]} \\
    &= |c| \cdot d(\mathbf v, 0) &\text{[absolute homogeneity]} \\
    &= c \cdot ||\mathbf v|| &\text{[def. metric-induced norm]} \\
    &= c &\text{[unit vector } v\text{]}
\end{align*}
\end{proof}

\subsection{Lipschitz functions in multiple dimensions}
\label{app:lipschitz}

Recall that if a function $f: \R \to \R$ is Lipschitz($\lambda$), then for any $x, a \in \R$:
\begin{equation*}
    |f(x) - f(a)| \leq \lambda |x-a|.
\end{equation*}
This implies that the function's derivatives are bounded by $\lambda$: we cannot grow faster than $\lambda$ as we go from $x$ to $a$.

In higher dimensions, $|x-a|$ is no longer scalar-valued;
as a result, Lipschitz functions must be defined with respect to a distance metric.
Formally speaking, we equip $\Rp$ with a distance metric $d(\cdot, \cdot)$ of the form given by Equation \ref{eq:l2dist} or \ref{eq:linfdist}.
Then $f: \Rp \to \R$ is Lipschitz($\lambda$) with respect to distance $d(\cdot, \cdot)$ if for any $\mathbf{x}, \mathbf{a} \in \mathbb{R}^n$:
\begin{equation*}
    |f(\mathbf x) - f(\mathbf  a)| \leq \lambda d(\mathbf  x, \mathbf{a}).
\end{equation*}
Notably, a function can only be Lipschitz \textbf{relative to a given distance metric}, so, e.g., a function that is Lipschitz$(\lambda)$ with respect to $L_2$ distance may not be Lipschitz$(\lambda)$ with respect to $L_\infty$ distance.

Multivariate Lipschitz functions have bounded derivatives like their unidimensional counterparts.
In particular, Lemma \ref{lem:lipbdsdd} shows that the directional derivatives of any Lipschitz function are bounded.
\begin{lemma}[Lipschitz bound on directional derivative]
\label{lem:lipbdsdd}
Suppose $f: \Rp \to \R$ is Lipschitz($\lambda$) with respect to $d(\cdot, \cdot)$. For any $\mathbf x \neq \mathbf a$, denote the unit vector \( \mathbf{v} = \frac{\mathbf{x}-\mathbf{a}}{d(\mathbf{x},\mathbf{a})} \). Then for direction $v$
\begin{equation*}
    \nabla_{\mathbf{v}} f(\mathbf{a}) \leq \lambda
\end{equation*}
if $\nabla_{\mathbf{v}} f(\mathbf{a})$ is defined.
\end{lemma}
\begin{proof}
\begin{align*}
    \nabla_{\mathbf{v}} f(\mathbf{a}) 
    &= \lim_{h \to 0} \frac{f(\mathbf a + h \mathbf v) - f(\mathbf a)}{h} &\text{[def. directional derivative]}\\
    &= \lim_{h \to 0} \frac{f(\mathbf a + h \mathbf v) - f(\mathbf a)}{d(\mathbf a + h \mathbf v, \mathbf a)} &\text{Lemma }\ref{lem:unitdist}\\
    &\leq \lambda. &\text{def. Lipschitz}
\end{align*}
\end{proof}

\subsection{Proof that CSM is in the MIB (Monotonic Imbalance Bounding) Class}
\label{app:mib}

\citet{iacus2011multivariate} defines the monotonic imbalance bounding (MIB) class of matching methods as follows:
\begin{definition}
\label{def:mib}
    A matching method is MIB for a function $f(\chi)$ of the dataset $\chi$ with respect to distance metric $d(\cdot, \cdot)$ if there exists some monotonically increasing function $\gamma_{f,d}(\cdot)$ on any non-negative p-dimensional vectors $\boldsymbol \pi \in \Rp_+$ such that:
    \begin{align}
        d(f(\chi_{m_T(\boldsymbol{\pi})}), f(\chi_{m_C(\boldsymbol{\pi})})) \leq \gamma_{f,d}(\boldsymbol{\pi}) \label{eq:MIB}
    \end{align} 
\end{definition}
Here, $f$ is function that summarizes a dataset.
$\boldsymbol{\pi} \in \R^p$ is the ($p$-dimensional) vector of tuning parameters. In our paper, $\boldsymbol{\pi}$ is the vector of covariate-wise calipers.  $\chi_{m_T(\boldsymbol{\pi})}$ and $\chi_{m_C(\boldsymbol{\pi})}$ are datasets of matched treated and matched control units, when we set our parameters to $\boldsymbol{\pi}$.
The function $\gamma_{f,d}(\cdot)$ is monotonically increasing if it is non-decreasing in every dimention of $\boldsymbol{\pi}$ and strictly increasing in at least one of them. 
This essentially says as we reduce any dimension of $\boldsymbol{\pi}$, the bound on $d(\cdot,\cdot)$ also decreases.

We show that CSM is MIB by choosing appropriate functions $f$, $d$, and $\gamma$ to verify Equation~\ref{eq:MIB}. We select $f$ to be the weighted mean for the $k$-th covariate, where the weight is obtained from matching. Specifically, we define $\bar{X}^{\mathbf{w}}_{Tk}$ and $\bar{X}^{\mathbf{w}}_{Ck}$ as the weighted means for the treated and control units after matching: 
\begin{align*}
 f(\chi_{m_T(\boldsymbol{\pi})}) &= \bar X^{\mathbf w}_{Tk} =\frac{1}{n_T} \sum_{t \in \mathcal{T}} X_{tk}\\
 f(\chi_{m_C(\boldsymbol{\pi})}) &= \bar X^{\mathbf w}_{Ck} = \frac{1}{n_T} \sum_{t \in \mathcal{T}} \sum_{j \in \mathcal{C}_t} w_{jt} X_{jk} .
\end{align*}

We further choose $d(x,y) = |x-y|$ and $\gamma_{f,d}(\boldsymbol{\pi}) = \gamma_k(\pi_k)$. Then Equation \ref{eq:MIB} becomes $|\bar{X}_{T k} - \bar{X}_{C k}| \leq \pi_k$. This inequality is established by Proposition \ref{prop:CSM-is-MIB}, hence CSM is MIB. 

We now provide a proof of Proposition \ref{prop:CSM-is-MIB}. First, we introduce a useful lemma:

\begin{lemma}
\label{prop:distmetriccal}
    Given covariatewise caliper $\boldsymbol{\pi} \in \Rp_{+}$ on units $t$ and $j$, for scaling matrix $V = diag\{\frac{1}{\boldsymbol{\pi}}\}$:
    \begin{enumerate}[label=(\alph*)]
        \item $d^{(2)}_V(\Xt, \Xj) \leq 1 \implies |X_{tk}-X_{jk}| \leq \pi_k \text{ for } k=1,\dots,p $
        \item $ d^{(\infty)}_V(\Xt, \Xj) \leq 1 \iff |X_{tk}-X_{jk}| \leq \pi_k \text{ for } k=1,\dots,p $
    \end{enumerate}
\end{lemma}
\begin{proof}
    \begin{enumerate}[label=(\alph*)]
        \item Without loss of generality, suppose for contradiction that $|X_{tk}-X_{jk}| > \pi_k \text{ for } k=1$. 
        Then:
        \begin{align*}
            d^{(2)}_V(\Xt,\Xj) 
            = \sqrt{\sum_{k=1}^p \frac{(X_{tk}-X_{jk})^2}{\pi_k^2}}
            \geq \sqrt{\frac{(X_{t1}-X_{j1})^2}{\pi_1^2}}
            > 1.
        \end{align*}
        \item This follows from definitions:
        \begin{align*}
            d^{(\infty)}_V(\Xt, \Xj) \leq 1
            &\iff \sup_{k = 1, \dots, p} |\frac{X_{tk} - X_{jk}}{\pi_k}| \leq 1 \\
            &\iff |X_{tk}-X_{jk}| \leq \pi_k \text{ for } k=1,\dots,p 
        \end{align*}
    \end{enumerate}
\end{proof}

\textit{Proof of Proposition \ref{prop:CSM-is-MIB}:}

Applying Lemma \ref{prop:distmetriccal} above, \begin{align*}
   |\bar X^{\mathbf w}_{T k}- \bar X^{\mathbf w}_{C k}| &=|\frac{1}{n_T} \sum_{t \in \mathcal{T}}( X_{tk} - \sum_{j \in \mathcal{C}_t} w_{jt} X_{jk}) | \\
   &= |\frac{1}{n_T} \sum_{t \in \mathcal{T}} \sum_{j \in \mathcal{C}_t} w_{jt}  (X_{tk} - X_{jk}) | &\text{Sum of } w_{jt} \text{ over }  \mathcal{C}_t \text{ is 1}  \\
   &\leq \frac{1}{n_T} \sum_{t \in \mathcal{T}} \sum_{j \in \mathcal{C}_t} w_{jt} | X_{tk} - X_{jk} | &\text{Triangle inequality} \\
    & \leq \frac{1}{n_T} \sum_{t \in \mathcal{T}} \sum_{j \in \mathcal{C}_t} w_{jt} \pi_k &\text{Lemma \ref{prop:distmetriccal}} \\
    & = \pi_k  \left( \frac{1}{n_T} \sum_{t \in \mathcal{T}} \sum_{j \in \mathcal{C}_t} w_{jt}  \right) \\
    &= \pi_k 
\end{align*}

The last line is because $ \frac{1}{n_T} \sum_{t \in \mathcal{T}} \sum_{j \in \mathcal{C}_t} w_{jt} = \frac{1}{n_T} \sum_{t \in \mathcal{T}} 1 = 1 $

\subsection{Proof of Proposition \ref{prop:wass}}
\label{app:wass}

We prove the bound on the Wasserstein distance by choosing a specific coupling:
\begin{proof}
    For :
    \begin{align*}
    \mathcal{W}_q^{(\ell)} (f_T, f_C) 
    &= \inf_{\substack{\bX \sim f_T \\ \mathbf{Y} \sim f_C}} E\big[ d_V^{(\ell)}(\bX, \mathbf{Y})^q \big]^{1/q}
    \end{align*}
    
    We choose a coupling, generated as:
    \begin{enumerate}
        \item Sample $\bX \sim f_{T1}$ as $\bX \sim \text{Uniform}(\{\bX_1, \dots, \bX_{n_T}\})$, so $\bX = \Xt$ for some $t \in 1, \dots, n_T$.
        \item Sample $\mathbf{Y} \sim f_{C1} \mid \bX = \Xt$ as $\mathbf{Y} \sim \text{Weighted Uniform}(\{\bX_j : j \in \Ct\})$ for the control units matched to treated unit $t$, with their appropriate weights.
    \end{enumerate}
    This coupling ($f_{T1}, f_{C1}$) clearly produces the correct marginals ($f_{T}, f_{C}$), so we can write:
    \begin{align*}
    \mathcal{W}_q(f_T, f_C)^q
    &\leq E_{\substack{X \sim f_{T1} \\ Y \sim f_{C1}\mid X_t}} \big[ d(X, Y)^q \big] \\
    &= E_{X \sim f_{T1}} \Big[ E_{Y \sim f_{C1} \mid X} \big[ d(X,Y)^q \mid X \big] \Big] \\
    &= \frac{1}{n_T} \sum_{t \in \mathcal{T}} \Big[ \sum_{j \in \Ct} w_{jt} d(X_t, X_j)^q \Big] \\
    &\leq c^q \cdot \frac{1}{n_T} \sum_{t \in \mathcal{T}} \Big[ \sum_{j \in \Ct} w_{jt} \Big] \\
    &= c^q.
    \end{align*}
For $q = \infty$, take the limit as $q$ increases of the above result.
\end{proof}

\subsection{Proof of Proposition \ref{prop:meanbd}}
\label{app:meanbd}

Proposition \ref{prop:meanbd}, showing the bound on the covariate mean differences, is trivially proved using the fact that $d^{(2)}_V(\cdot, \cdot)$ and $d^{(\infty)}_V(\cdot, \cdot)$ are induced norms (see Appendix \ref{app:distmetrics}).

\begin{proof}
    \begin{align*}
        d_V(\bar{\bX}_T, \Bar{\bX}_C)
        &= d_V(\frac{1}{n_T} \sum_t \bX_t, \frac{1}{n_T} \sum_t \sum_{j \in \Ct} w_{jt} \bX_j) \\
        &= d_V(0, \frac{1}{n_T} \sum_t \sum_{j \in \Ct} w_{jt} \bX_j - \frac{1}{n_T} \sum_t \bX_t) &\text{[translation invariance]} \\
        &= d_V(0, \frac{1}{n_T} \sum_t \sum_{j \in \Ct} w_{jt} (\bX_j - \bX_t)) &[\sum_j w_{jt}=1] \\
        &= ||\frac{1}{n_T} \sum_t \sum_{j \in \Ct} w_{jt} (\bX_j - \bX_t)||_V &[||\cdot||_v \text{ induces } d_V(\cdot, \cdot)] \\
        &\leq \frac{1}{n_T} \sum_t \sum_{j \in \Ct} w_{jt} 
            ||\bX_j - \bX_t||_V &[\text{triangle inequality}] \\
        &= \frac{1}{n_T} \sum_t \sum_{j \in \Ct} w_{jt}
            d_V(\bX_j, \bX_t) \\
        &\leq \frac{1}{n_T} \sum_t \sum_{j \in \Ct} w_{jt} c & [d_V(\bX_j, \bX_t) \leq c] \\
        &= c
    \end{align*}
\end{proof}

\subsection{Derivations and technical details of the synthetic controls approach}
\label{app:scm}

Synthetic controls naturally control linear bias, due to the following observations.

\begin{proposition}
\label{prop:scm_is_projection}
Synthetic control weights project the treated unit's covariates onto the convex hull of the donor pool units' covariates.
\end{proposition}
\begin{proof}
    Recall that the convex hull of $X = \{\mathbf{X}_1, \dots, \mathbf{X}_J\}$ for $\Xj \in \Rp$ is:
    \begin{equation*}
        conv(X) = \big\{\sum_{j=1}^J w_{jt} \Xj \mid \sum w_{jt} = 1, w_{jt} \geq 0 \text{ for } j = 1, \dots, J \big\}
    \end{equation*}
 See, e.g., \citet{boyd2004convex}.
    
    Recall also that the Euclidean projection of $\Xt$ onto the set $conv(X)$ in the norm $||\cdot||$ is defined as:
    \begin{align*}
        \argmin_\mathbf{u} &\hspace{2mm} ||\Xt - \mathbf{u}|| \\
        \text{subject to} &\hspace{2mm} \mathbf{u} \in conv(X)
    \end{align*}
    
    This optimization problem is equivalent to standard SCM.
\end{proof}
The covariates of the donor pool generally form a $p$-dimensional convex hull with a $p$ dimensional covariate,\footnote{Assuming that there are at least $p+1$ non-collinear donor-pool units. Otherwise, the dimension of the convex hull may be less than $p$.} i.e., a $p$-dimensional polygon (i.e., polytope) that contains the points corresponding to each of the donor pool units' covariates, as well as the lines between those points.
Geometrically, synthetic controls simply find the point in this convex hull that is ``closest'' to the treated unit's covariates (i.e., ``project'' the treated unit onto the convex hull), where ``closeness'' is defined using the given distance metric.
Importantly, every point in a convex hull can be written as a convex weighted average of the points generating the hull, so the ``closest'' point corresponds to a set of non-negative donor-pool-unit weights that sums to one.

To impute the counterfactual outcome for the treated unit (i.e., the outcome for the constructed synthetic control unit), synthetic controls linearly interpolate the donor-pool units' outcomes to the projected point.
Slightly more formally:
\begin{remark}
\label{rem:lin_interp}
    Suppose we use $Y_t^* \equiv \sum_{j \in \Ct} w_{jt} Y_j$ as our estimate of the counterfactual outcome for treated unit $t$, using control units $j$ with convex weights $w_{jt}$.
    Then we may interpret $Y_t^*$ as a linear interpolation of the control units' outcomes to $\sum_{j \in \Ct} w_{jt} \Xj$.
\end{remark}
\begin{proof}
    Write $Y_j = f(\Xj)$, without noise.
    If we suppose that $f(\cdot)$ is linear, by the definition of a linear map we have:
    \begin{align*}
        \sum_{j \in \Ct} w_{jt} Y_j
        &= \sum_{j \in \Ct} w_{jt} f(\Xj) \\
        &= f(\sum_{j \in \Ct} w_{jt} \Xj).
    \end{align*}
    I.e., $Y_j^*$ may be interpreted as the evaluation of $f(\cdot)$ at $\sum_{j \in \Ct} w_{jt} \Xj$.
    Convex weights ensure that this is interpolation, not extrapolation.
\end{proof}

In summary, synthetic controls linearly interpolate the donor-pool units' outcomes to the point on the convex hull closest to the treated unit.
As a result, they are subject to linear interpolation bias \citep{kellogg2021combining} in this first stage, though the bias is controlled by the maximum distance across which they are allowed to interpolate with the caliper.
Synthetic controls then flatly extrapolate this outcome to the treated unit, i.e., impute the treated unit's outcome as the linearly interpolated value.
This second step is subject to potential extrapolation bias from the point on the convex  hull to the location of the treated unit \citep{kellogg2021combining} proportional to $d_V(\Xt, \sum_{j \in \Ct} w_{jt} \Xj)$, assuming the potential outcome function is Lipschitz.
Both interpolation and extrapolation bias are therefore controlled by the caliper size in CSM.

\subsection{Optimization for finding SCMs}
\label{rem:linf_opt}

For a scaled $L_2$ distance, the SCM optimization is typically directly solved as a quadratic programming problem.
For a scaled $L_\infty$ distance, we have to transform the SCM optimization into the following linear programming problem, which we solve for each treated unit $t$ in turn:

\begin{align*}
    \text{minimize} &\hspace{2mm} y\\
    \text{subject to} &\hspace{2mm} -y \leq \Big[ V(\Xt - \sum_{j \in \Ct} w_{jt} \Xj) \Big]_k \leq y \text{ for } k = 1, \dots, p \\
    &\hspace{2mm} \sum_{j \in \Ct} w_{jt} = 1 \\
    &\hspace{2mm} 0 \leq w_{jt} \leq 1 \text{ for } j \in \Ct
\end{align*}

In words, we are trying to find the smallest $y \geq 0$ that ``pinches'' the covariate balance for each covariate.
The $w_{jt}$ are all free (under their constraints), so the goal of shrinking $y$ induces us to find $w_{jt}$ to make small $y$ achievable.

\subsection{Proof of bias reduction from synthetic controls (Proposition \ref{prop:scbiasbd})}
\label{app:scbiasbd}

Proving Proposition \ref{prop:scbiasbd} requires some additional mathematical exposition.
Specifically, we require $f_0(\cdot)$ to be differentiable in order to use a Taylor expansion.
That said, to effectively use the Lipschitz property, standard multivariable Taylor expansion does not suffice.
Instead, we require a Taylor expansion in a distance metric.

\begin{lemma}[Taylor expansion in a distance metric]
\label{lem:bias}
Suppose $f: \Rp \to \R$ is differentiable.
Let $\Xj, \Xt \in \Rp$, and write unit vector $\vj \equiv \frac{\Xj - \Xt}{d_V(\Xj, \Xt)}$ where  $d_V$ is $d_V^{(2)}$ or $d_V^{(\infty)}$, given by Equations \ref{eq:l2dist} or \ref{eq:linfdist}.
Then:
$$f(\Xj) = f(\Xt) + d_V(\Xj, \Xt) \nabla_{\vj} f(\Xt) + o(d_V(\Xj, \Xt))$$
where $\nabla_{\vj} f \equiv \nabla f \cdot \vj$ is the the directional derivative of $f$ in direction $\vj$
\end{lemma}
\begin{proof}
    By standard multivariate Taylor expansion, we know that:
    \begin{align*}
        f(\Xj) 
        &= f(\Xt) + (\Xj - \Xt)^T \nabla f(\Xt) + o(||\Xj - \Xt||)
    \end{align*}
    for the usual Euclidean norm $||\cdot||$.

    We then have, using our definition of our directional derivative 
\begin{align*}
    (\Xj - \Xt)^T \nabla f(\Xt) &= d_V(\Xj, \Xt)\frac{1}{d_V(\Xj, \Xt)} (\Xj - \Xt)^T \nabla f(\Xt) \\
     &= d_V(\Xj, \Xt) \nabla_{\vj} f(\Xt)
\end{align*}
    which gives
    \begin{align*}
        f(\Xj)
        &= f(\Xt) + d_V(\Xj, \Xt) \nabla_{\vj} f(\Xt) + o(||\Xj - \Xt||)
    \end{align*}
    
    Finally, for the little-$o$ remainder term, we show that if $g_a(x) = o(||x-a||)$, then $g_a(x) = o(d_V(x,a))$. We demonstrate this by proving that $\lim_{x \to a} \frac{g_a(x)}{d_V(x,a)} = 0$.
    \begin{align*}
        \lim_{x \to a} \frac{g_a(x)}{d_V(x,a)}
        &= \lim_{x \to a} \frac{g_a(x)}{||x-a||} \cdot \frac{||x-a||}{d_V(x,a)} \\ 
        &= \lim_{x \to a} \frac{g_a(x)}{||x-a||} \cdot \lim_{x \to a} \frac{||x-a||}{d_V(x,a)},
    \end{align*}
    if both limits exist.
    We know $\lim_{x \to a} \frac{g_a(x)}{||x-a||} = 0$ since $g_a(x) = o(||x-a||)$, so it remains to show that $\lim_{x \to a} \frac{||x-a||}{d_V(x,a)}$ exists.
    \begin{align*}
        \lim_{x \to a} \frac{||x-a||}{d_V(x,a)}
        &= \lim_{t \to 0} \frac{t}{d_V(a+tv,a)} &[tv=x-a, t = ||x-a||]\\
        &= \lim_{t \to 0} \frac{1}{\ddt d_V(a+tv, a)} &[\text{L'Hospital}] 
    \end{align*}
 The above limit exists because $\ddt d_V(a+tv, a) = \ddt d_V(tv, 0)  = \ddt t ||v||_V = ||v||_V$ is non-zero when $v \neq 0$.
 Hence $g_a(x) = o(d_V(x,a))$.
\end{proof}

Lemma \ref{lem:bias} is technical, but it captures a very simple intuition.
Standard multivariate Taylor expansion takes the dot product of the $p$-dimensional gradient of $f(\cdot)$ with the $p$-dimensional vector of differences between $\mathbf{x}$ and the point $a$.
Calipers, however, only control the scalar quantities $d(\Xt, \Xj)$.
Lemma \ref{lem:bias} simply rewrites standard multivariate Taylor expansion to better utilize the fact that Lipschitz functions control scalar directional derivatives (as discussed in Appendix \ref{app:lipschitz}).

With our lemma, we now provide the proof for Proposition \ref{prop:scbiasbd},
where we use the fact that, by Lemma \ref{lem:lipbdsdd}, Lipschitz functions have bounded directional derivatives.
\begin{proof}
    Recall that by Taylor expansion of $f_0(\cdot)$ around $\Xt$, we have:
    \begin{align*}
        \sum_j w_{jt} f_0(\Xj) - f_0(\Xt)
        &=  \sum_j w_{jt} d_V(\Xj, \Xt) \nabla_{\vj} f_0(\Xt) + \sum_j w_{jt} o(d_V(\Xj, \Xt)) \\
        &= \sum_j w_{jt} d_V(\Xj, \Xt) \nabla_{\vj} f_0(\Xt) + o(d_V)
    \end{align*}
    
    Consider the linear term:
    \begin{align*}
        \sum_j w_{jt}  d_V(\Xj, \Xt) \nabla_{\vj} f_0(\Xt)
        &= \sum_j w_{jt} \nabla f_0(\Xt)^T (\Xj - \Xt) &[\text{def. } \nabla_{\vj}]\\
        &= \sum_j w_{jt} \sum_{k=1}^p c_k (\mathbf{X}_{jk} - \mathbf{X}_{tk}) \\
        &= \sum_{k=1}^p c_k (\sum_j w_{jt} \mathbf{X}_{jk} - \mathbf{X}_{tk}) \\
        &= 0 &[\text{exact SC}]
    \end{align*}
    where $\mathbf{c} = [c_1 \dots c_p]^T$ is the fixed (unknown) gradient of $f_0(\cdot)$ at $\Xt$.
    Thus, if the synthetical control matches, we drop the linear term from our bias.

    Using the above we have, given that, for all $t$, $d_V(\Xt, \Xj) \leq c$ for all $j \in \Ct$:
    
    \begin{align*}
 \sum_t \frac{1}{N_t} \sum_j w_{jt} f_0(\Xj) - f_0(\Xt) =   \frac{1}{N_t}  \sum_t \sum_j w_{jt} o( d_V(X_j, X_t) ) =  \frac{1}{N_t} \sum_t \sum_j w_{jt} o( c ) = o( c ) .
    \end{align*}
\end{proof}

\end{document}